\documentclass{amsart}

\usepackage{fullpage}
\usepackage{xspace}
\usepackage{enumerate}
\usepackage{amsmath,amsthm,amsfonts}
\usepackage{bm}
\usepackage{times}
\usepackage{ifthen}
\usepackage{nicefrac}
\usepackage[colorlinks=true,
linkcolor=webgreen,
filecolor=webbrown,
citecolor=webgreen]{hyperref}
\usepackage{url}
\usepackage{algorithm}
\usepackage[noend]{algpseudocode}
\usepackage[capitalize]{cleveref}
\usepackage{tikz}
\usetikzlibrary{arrows.meta}

\theoremstyle{plain}
\newtheorem{theorem}{Theorem}
\newtheorem*{theorem*}{Theorem}

\newtheorem*{corollary*}{Corollary}
\newtheorem{lemma}[theorem]{Lemma}
\newtheorem*{lemma*}{Lemma}
\newtheorem{proposition}[theorem]{Proposition}
\newtheorem*{proposition*}{Proposition}
\theoremstyle{definition}
\newtheorem{definition}{Definition}
\newtheorem*{definition*}{Definition}
  
\newtheorem*{example*}{Example}

\theoremstyle{remark}
\newtheorem{remark}[theorem]{Remark}
\newtheorem*{remark*}{Remark}

\newtheorem*{notation*}{Notation}
\newcommand{\NN}{{\mathbb N}}
\newcommand{\RR}{{\mathbb R}}
\newcommand{\OPT}{\operatorname{OPT}}

\definecolor{webgreen}{rgb}{0,.5,0}
\definecolor{webbrown}{rgb}{.6,0,0}
\providecommand{\half}{\ensuremath{\nicefrac{1}{2}}\xspace}

\providecommand{\SetCard}[1]{\ensuremath{\left| #1 \right|}\xspace}
\providecommand{\SET}[1]{\ensuremath{\left\{ #1 \right\}}\xspace}
\providecommand{\Abs}[1]{\ensuremath{\left| #1 \right|}\xspace}
\providecommand{\Set}[2]{\ensuremath{\SET{#1 \mid #2}}\xspace}

\providecommand{\Kth}[1]{\ensuremath{{#1}^{\rm th}}}
\providecommand{\Ceiling}[1]{\ensuremath{\left\lceil {#1} \right\rceil}\xspace}
\providecommand{\Floor}[1]{\ensuremath{\left\lfloor {#1} \right\rfloor}\xspace}

\providecommand{\PROB}{\ensuremath{{\rm Prob}}\xspace}
\providecommand{\Prob}[2][]{\ensuremath{%
\ifthenelse{\equal{#1}{}}{\PROB[#2]}{\PROB_{#1}[#2]}}\xspace}

\providecommand{\Expect}[2][]{\ensuremath{%
\ifthenelse{\equal{#1}{}}{\mathbb{E}}{\mathbb{E}_{#1}}%
\left[#2\right]}\xspace}

\newcommand{\BT}{\ensuremath{T}\xspace}
\newcommand{\BV}{\ensuremath{V}\xspace}
\newcommand{\BE}{\ensuremath{E}\xspace}
\newcommand{\ST}{\ensuremath{\mathcal{T}}\xspace}
\newcommand{\STP}{\ensuremath{\mathcal{T}'}\xspace}

\newcommand{\NumBV}{\ensuremath{n}\xspace}

\newcommand{\QUERY}{\ensuremath{q}\xspace}
\newcommand{\Query}[1]{\ensuremath{\QUERY_{#1}}\xspace}
\newcommand{\SEQ}[1][]{\ensuremath{\ifthenelse{\equal{#1}{}}{Q}{Q^{#1}}}\xspace}
\newcommand{\SEQH}[1][]{\ensuremath{\ifthenelse{\equal{#1}{}}{\hat{Q}}{\hat{Q}^{#1}}}\xspace}
\newcommand{\Seq}[2]{\ensuremath{\SEQ[](#1,#2)}\xspace}

\newcommand{\GENCOST}[1][]{\ensuremath{\ifthenelse{\equal{#1}{}}{f}{f_{#1}}}\xspace}
\newcommand{\GenCost}[3][]{\ensuremath{\GENCOST[#1]\left(#2,#3\right)}\xspace}
\newcommand{\COST}[1][]{\ensuremath{\ifthenelse{\equal{#1}{}}{h}{h^{#1}}}\xspace}
\newcommand{\Cost}[2][]{\ensuremath{\COST[#1]\left(#2\right)}\xspace}
\newcommand{\SeqCost}[1]{\ensuremath{c\left(#1\right)}\xspace}
\newcommand{\AlgCost}[1]{\ensuremath{C\left(#1\right)}\xspace}
\newcommand{\TreeCost}[2]{\ensuremath{C_{#1}\left(#2\right)}\xspace}

\newcommand{\DB}{\ensuremath{p}\xspace} % polynomial degree
\newcommand{\CB}{\ensuremath{s}\xspace} % bound on coefficients
\newcommand{\PolCo}[2][]{\ensuremath{\ifthenelse{\equal{#1}{}}{\alpha_{#2}}{\alpha^{#1}_{#2}}}\xspace} % polynomial coefficients
\newcommand{\CumCo}[2][]{\ensuremath{\ifthenelse{\equal{#1}{}}{\beta_{#2}}{\beta^{#1}_{#2}}}\xspace} % cumulative cost coefficients
\newcommand{\CumCoT}[2][]{\ensuremath{\ifthenelse{\equal{#1}{}}{\gamma_{#2}}{\gamma^{#1}_{#2}}}\xspace}
\newcommand{\EnCo}[2]{\ensuremath{\beta_{#1}^{#2}}\xspace}

\newcommand{\DIST}[1][]{\ensuremath{\ifthenelse{\equal{#1}{}}{d}{d_{#1}}}\xspace}
\newcommand{\Dist}[3][]{\ensuremath{\DIST[#1](#2,#3)}\xspace}
\newcommand{\Boundary}[2][]{\ensuremath{\ifthenelse{\equal{#1}{}}{\partial\left(#2\right)}{\partial_{#1}\left(#2\right)}}\xspace}
\newcommand{\Path}[2][]{\ensuremath{\ifthenelse{\equal{#1}{}}{P\left(#2\right)}{P_{#1}\left(#2\right)}}\xspace}
\newcommand{\ch}[2][]{\ensuremath{\ifthenelse{\equal{#1}{}}{\text{ch}\left(#2\right)}{\text{ch}_{#1}\left(#2\right)}}\xspace}

\newcommand{\ALG}{\ensuremath{\mathcal{A}}\xspace}
\newcommand{\BS}{\ensuremath{\mathcal{B}}\xspace}
\newcommand{\TV}{\ensuremath{t}\xspace} % target vertex
\newcommand{\LAB}{\ensuremath{\lambda}\xspace} % node label
\newcommand{\Lab}[1]{\ensuremath{\LAB\left(#1\right)}\xspace}
\newcommand{\Feas}[2][]{\ensuremath{\ifthenelse{\equal{#1}{}}{\phi\left(#2\right)}{\phi_{#1}\left(#2\right)}}\xspace}
\newcommand{\IFV}[1][]{\ensuremath{\ifthenelse{\equal{#1}{}}{\bar{v}}{\bar{v}_{#1}}}\xspace}
\newcommand{\IFS}[1][]{\ensuremath{\ifthenelse{\equal{#1}{}}{Q}{Q_{#1}}}\xspace}

\begin{document}

\begin{titlepage}
\title{Binary Search with Distance-Dependent Costs}

\author{Calvin Leng}
\address{Calvin Leng, University of Southern California, Los Angeles, CA 90089, US}
\email{cleng@usc.edu}

\author{David Kempe}
\address{David Kempe, University of Southern California, Los Angeles, CA 90089, US}
\email{david.m.kempe@gmail.com}

\maketitle

\begin{abstract}
We introduce a search problem generalizing the typical setting of Binary Search on the line.
Similar to the setting for Binary Search, a target is chosen adversarially on the line, and in response to a query, the algorithm learns whether the query was correct, too high, or too low.
Different from the Binary Search setting (which is the special case of a constant function), the cost of a query is a monotone non-decreasing function of the distance between the query and the correct answer; different functions can be used for queries that are too high vs.~those that are too low.
The algorithm's goal is to identify an adversarially chosen target with minimum total cost.
Note that the algorithm does not even know the cost it incurred until the end, when the target is revealed.
This abstraction captures many natural settings in which a principal experiments by setting a quantity (such as an item price, bandwidth, tax rate, medicine dosage, etc.) where the cost or regret increases the further the chosen setting is from the optimal one.

First, we show that for arbitrary symmetric cost functions (i.e., the case when overshooting vs.~undershooting by the same amount leads to the same cost), the standard Binary Search algorithm is a 4-approximation.
  
We then show that when the cost functions are bounded-degree polynomials of the distance (including the asymmetric case), the problem can be solved optimally using Dynamic Programming; this relies on a careful encoding of the combined cost of past queries (which, recall, will only be revealed in the future).
We then generalize the setting to finding a node on a tree; here, the
response to a query is the direction on the tree in which the target is located, and the cost is increasing in the distance on the tree from the query to the target.
Using the $k$-cut search tree framework of Berendsohn and Kozma and the ideas we developed for the case of the line, we give a PTAS when the cost function is a bounded-degree polynomial.
\end{abstract}
\end{titlepage}

\section{Introduction}\label{s:intro}
Binary Search is one of the most fundamental algorithms, with applications extending far beyond computer science. Even children usually figure out Binary Search at an early age via a ``Number Guessing'' game in which the feedback in response to a guess is only whether the guess was correct, too high, or too low. Without additional distributional assumptions on the correct answer, Binary Search minimizes the worst-case number of guesses until the correct number is guessed.

Applying this insight to real-world situations, Binary Search will lead to the fastest convergence when an unknown quantity must be guessed in a setting in which only feedback of the form ``too high''/``too low'' is provided.
Natural applications abound, such as determining the right medicine dosage, price to charge for an item, maximum network bandwidth, setting for a thermostat, tax rate, and many others.
Inspecting these examples, we see that not only does it matter how long it takes until the right value is found, but there is a cost to guesses that are far from the correct value. A significantly wrong medicine dosage can have more deleterious health effects, an item price that is much too low loses more potential revenue, a room that is much too cold/hot will be less comfortable than one that is only slightly so, etc.

This suggests a natural refinement of the standard Binary Search model, in which non-uniform costs are associated with queries.\footnote{The standard model corresponds to a cost of 1 for each incorrect query.} While a substantial literature --- discussed in \cref{s:related-work} --- has considered non-uniform query costs, these costs have mostly been associated with the query itself, and were independent of the correct answer. In the thermostat example, this would correspond to different costs of setting the thermostat to 58F vs.~68F, irrespective of what the ideal temperature is. While there are many (different) applications for such a model of query-dependent costs, a model in which the costs depend on the distance between the query and the correct answer is --- as discussed above --- at least equally natural.

Two works have studied online optimization problems in which the incurred cost depends on both the query and the ideal setting. Karp et al.~\cite{karp2000congestion} study congestion control in networks. Here, the goal is to determine the ideal bandwidth at which to send data. Sending at higher bandwidth causes a cost due to packet loss, while sending at a lower bandwidth results in opportunity cost equal to the difference between the ideal and chosen bandwidth. Kleinberg and Leighton \cite[Section 2]{kleinberg2003oppa} study single-item pricing for repeated interactions with the same single buyer. Here, setting the price too high causes the buyer to not buy at all, while setting the price too low causes opportunity cost equal to the difference between the buyer's valuation and the chosen price.

One difficulty in this setting is the following: the algorithm does not know how much cost it incurred until it terminates and observes the correct answer. This distinguishes the setup from that of many optimization problems.

\subsection{The Basic Setup}
We consider the following basic setup, described in detail in \cref{s:prelims}. The algorithm is to identify an adversarially chosen number $t$ from $\SET{1, \ldots, \NumBV}$. To do so, it repeatedly queries numbers \Query{i} from $\SET{1, \ldots, \NumBV}$. In return, it learns whether $\Query{i} = \TV$ (in which case the process terminates), $\Query{i} < \TV$, or $\Query{i} > \TV$.
When making the query \Query{i} in round $i$, the algorithm incurs a cost of \GenCost{\Query{i}}{\TV}.
Here, \GENCOST is a known function given to the algorithm beforehand; however, the algorithm cannot evaluate the incurred cost until the end, when it learns \TV.
Since \TV is chosen adversarially, the algorithm's goal is to minimize the maximum total cost $\max_{\TV} \sum_i \GenCost{\Query{i}}{\TV}$ over all choices of \TV.

\subsection{Results for the Line}
First, in \cref{s:binary}, we show that on the line, Binary Search is a constant-factor approximation.

\begin{theorem} \label{t:binary-search}
  In the setting of searching on the line, assume that the cost function is monotone and symmetric\footnote{The assumption of symmetry is essential for the guarantee to hold. For example, if $\QUERY < \TV$ resulted in cost (essentially) 0, while $\QUERY > \TV$ always incurred high cost, then Linear Search would perform much better than Binary Search.} (i.e., $\GenCost{\QUERY}{\TV} = \Cost{\Abs{\QUERY-\TV}}$ depends only on the distance between \TV and \QUERY, but not on whether $\TV < \QUERY$ or $\TV > \QUERY$, and \COST is monotone non-decreasing).
  Then, the standard Binary Search algorithm, which always queries a midpoint of the remaining feasible set, is a 4-approximation.
\end{theorem}

We emphasize that this result holds for arbitrary monotone functions of distance, a larger class of functions than for our other results.
Also in \cref{s:binary}, we show that there a constant $c > 1$ such that Binary Search does not provide a $c$-approximation, even for linear cost functions. Determining a tight approximation guarantee of Binary Search is an interesting direction for future work.

Our second main result for the line, proved in \cref{s:paths}, is that when the cost function is a possibly asymmetric low-degree polynomial, an optimal search strategy can be found in polynomial time.

\begin{theorem} \label{t:path}
  In the setting of searching on the line, assume that the cost function is monotone, possibly asymmetric, and is a constant-degree polynomial.
  In other words, $\GenCost{\QUERY}{\TV} = \Cost{\QUERY-\TV}$ where $\Cost{x} = \Cost[-]{x}$ when $x < 0$ and $\Cost{x} = \Cost[+]{x}$ when $x > 0$, and $\COST[-], \COST[+]$ are non-decreasing polynomials of degree \DB.
  Then, there is an algorithm with running time $O(\NumBV^{\DB^2+3\DB+3})$ which finds a search strategy minimizing the maximum total incurred cost.
\end{theorem}

The key insight for proving \cref{t:path} is that all the relevant information about the history of past queries can be encoded sufficiently succinctly to enable a dynamic program. This encoding specifically draws on the fact that the cost function is a constant-degree polynomial, and can be considered akin to sketches, although in a different context from the streaming one that typically draws on sketches.

The fact that the cost depends on \TV \emph{only} through its distance from \QUERY is somewhat restrictive; in particular, it means that Theorem~\ref{t:path} does not apply to the regret in a repeated posted-price setting \cite{kleinberg2003oppa} (where the regret when $\QUERY > \TV$ is the constant \TV) or to the congestion control setting \cite{karp2000congestion} (where the opportunity cost when $\QUERY > \TV$ is also sometimes assumed to be the constant \TV).
We therefore provide the following generalization to bounded-degree polynomials of both \QUERY and \TV, the proof of which can be found in \cref{s:position-dependent}.
Here, to obtain a polynomial runtime, we need to assume that the \emph{coefficients} of the cost functions are also bounded by a polynomial in \NumBV.

\begin{theorem} \label{t:multivar_path}
In the setting of searching on the line, assume that the cost function \GENCOST is of the form $\GenCost{\QUERY}{\TV} = \Cost{\QUERY, \TV}$ for a\footnote{Again, we allow for different functions \COST[-], \COST[+] depending on whether $\QUERY < \TV$ or $\QUERY > \TV$.} bivariate polynomial \COST of maximum degree \DB whose coefficients are bounded in magnitude by $O(\NumBV^{\CB})$ for some constant \CB.
  Then, there is an algorithm with running time $O(n^{3 + (p+1) \cdot (p + 2s + 2)/2})$ which finds a search strategy minimizing the total incurred cost against an adversary.
\end{theorem}

\subsection{A Generalization to Trees} 

A straightforward alternative view of Binary Search is to think of $\SET{1, \ldots, \NumBV}$ as comprising the nodes of a path, which in turn is a special case of an undirected tree. Then, in response to a query \Query{i}, the algorithm learns which of the subtrees rooted at children of \Query{i} contains the target node \TV.
Indeed, for the standard search cost model (aiming to minimize the maximum number of queries), there has been a lot of work generalizing results from the line to search for an unknown node in a tree \cite{ben-asher:farchi:newman,cicalese:jacobs:laber:molinaro,cicalese:jacobs:laber:valentim,dereniowski:edge-ranking,dereniowski2017weightedtree,mozes:onak:weimann,onak:parys:generalization,cicalese2009avgtree}.

In this case, the algorithm is given an undirected tree \BT on \NumBV nodes as well as a cost function \COST. An adversary chooses a node \TV of the tree.
The algorithm performs a sequence of queries $\Query{1}, \Query{2}, \ldots$.
In response to a query \Query{i}, unless $\Query{i} = \TV$ (in which case the process terminates), the algorithm learns which of the neighbors of \Query{i} is the unique neighbor strictly closer to \TV in the tree.
The cost incurred by the algorithm with this query is $\GenCost{\Query{i}}{\TV}$.
When \GENCOST is unrestricted, this problem is NP-hard even when \BT is a spider; this follows from a trivial reduction from the weighted vertex query setting (with cost $w_q$ for querying node $q$) by simply setting $\GenCost{q}{\TV} = w_q$ for all \TV.
NP-hardness for the weighted vertex query setting was proved in \cite{cicalese2014weightedtree}.
Thus, we will focus on the case where $\GenCost{\QUERY}{\TV} = \Cost{\Dist[\BT]{\QUERY}{\TV}}$ where \Dist[\BT]{\QUERY}{\TV} is the length of the unique path in \BT between \TV and \QUERY, and \COST is a non-decreasing function.
As before, the algorithm does not observe the incurred cost until the target \TV is revealed.
For this setting, the ideas employed for the path generalize to yield a polynomial-time approximation scheme (PTAS)\footnote{We currently do not know if the restricted problem is NP-hard to solve optimally.}:

\begin{theorem} \label{t:tree}
  In the setting of searching on a tree \BT with \NumBV nodes, assume that the cost function \COST is monotone, and is a constant-degree polynomial.
  In other words, $\GenCost{\QUERY}{\TV} = \Cost{\Dist[\BT]{\QUERY}{\TV}}$ where \COST is a non-decreasing polynomial of degree \DB.
Then, for every $\epsilon > 0$, there is an algorithm with running time $O(\NumBV^{2+(\DB^2/2+3\DB/2+2) \cdot (1+2/\epsilon)})$ which finds a search strategy whose maximum total incurred cost is within a factor of $(1+\epsilon)$ of the optimum.
\end{theorem}

The algorithm is based on the same concise encoding of the history of past queries as was used for the case of the path. In addition, in order to avoid an exponential blowup in the number of subtrees to consider, it focuses the search only on $k$-cut search trees, as defined in \cite{berendsohn2020stt}: search trees with feasible sets consisting of only subtrees which are separated from their complement by at most $k$ edges. 
For constant $k$, the number of such subtrees is only polynomial in \NumBV, and we prove that the best search strategy based only on such $k$-cut trees is a $(1+2/k)$-approximation to the best overall search strategy.

\subsection{Related Work} \label{s:related-work}
Most closely related to our work are the above-mentioned papers by Karp et al.~\cite{karp2000congestion} and Kleinberg and Leighton \cite{kleinberg2003oppa}.
Both are concerned with specific objective functions arising in the respective contexts of congestion control and single item pricing.
In both cases, choosing a value that is too low leads to an opportunity cost, which is simply the difference between the available and selected bandwidth or between the buyer's valuation and the charged price.
When the value is too high, different costs are considered: for item pricing and a ``severe'' version of congestion control, the seller/sender receives nothing, so that the cost is the ground truth value.
In a more ``benign'' version of congestion control, the cost is $\frac{A}{B} \cdot (\QUERY-\TV)$ for some integers $A, B$.
The main goal in \cite{karp2000congestion} and \cite{kleinberg2003oppa} is to study the asymptotics of the incurred cost/regret (in \cite{kleinberg2003oppa} also under different assumptions about the buyer).
However, as part of their investigation, Kleinberg and Leighton \cite{kleinberg2003oppa} show asymptotically optimal strategies in the setting of interacting with the same buyer repeatedly, and Karp et al.~\cite{karp2000congestion} explicitly give dynamic programs for the objective functions we state above, i.e., degree-1 polynomials.
Thus, Theorem~\ref{t:path} and Theorem~\ref{t:multivar_path} can be viewed as generalizing part of the work of Karp et al.~in \cite{karp2000congestion} to a much larger class of cost functions.
The analysis in \cite{karp2000congestion} gives running times of $O(\NumBV^4)$ (for the severe setting) and $O(\NumBV^4 \cdot (A+B))$ (for the benign setting).
The dependence on \NumBV is better than that claimed in our theorems, but this is mostly due to a simplified and slightly suboptimal analaysis in our work, promoting simplicity of the analysis over tight exponents.
Note that we avoid the pseudo-polynomial factor $(A+B)$ in our Theorem~\ref{t:path} --- it arises in the analysis of Karp et al.~\cite{karp2000congestion} due to using summaries in line with those from our Theorem~\ref{t:multivar_path}.

In the literature, there have been quite a few generalizations of the search setup underlying Binary Search.
One of the first is to locate a target ideal of a partially ordered set $P$.
In full generality, this problem is known to be NP-hard \cite{carmo2004randomposet}, and thus, specific types of posets and average-case models \cite{carmo2004randomposet,dereniowski:edge-ranking,linial1983search,linial1985central,linial:saks:searching} have been studied.
Of particular interest to us is the case when the poset is tree-like.
A study of this model was initiated by Dereniowski in \cite{dereniowski:edge-ranking} where he noted that the poset search problem on tree-like posets is equivalent to the search problem on trees with edge queries $e$ that reveal the component of $T \setminus e$ containing the target.

This equivalence, in combination with prior work on optimal edge rankings by Lam and Yue \cite{lam:yue:edge-ranking}, yielded linear-time algorithms for both the edge and vertex search problem on trees. The restriction to trees (or other simpler graphs) is essential, in that Lam and Yue \cite{lam:yue:edge-ranking-hard} also show that the edge and vertex ranking problems on general graphs are NP-hard.
Polynomial-time algorithms for both the vertex and edge ranking problems on trees were then rediscovered, initially with an $O(n^4)$ time algorithm by Ben-Asher, Farchi, and Newman in \cite{ben-asher:farchi:newman}, improved to $O(n^3)$ time by Onak and Parys \cite{onak:parys:generalization}, and finally to linear time by Mozes, Onak, and Weimann \cite{mozes:onak:weimann} using similar methods as in the earlier work on edge-ranking.

There has also been a wide variety of work on weighted and average-case variants of the problem on paths \cite{frankl1993application,knight1988path,navvarro2000path} and trees \cite{cicalese2011weightedtree,cicalese2014weightedtree,dereniowski2017weightedtree,cicalese2009avgtree}.
In particular, the weighted edge-query problem can easily be reduced to the weighted node-query variant \cite{cicalese2011weightedtree}, and it is known that the weighted edge-query problem is NP-complete even on trees of diameter up to 4, trees of degree up to 16 \cite{cicalese2009avgtree}, and spiders \cite{cicalese2014weightedtree}.
Approximation algorithms were first provided by Cicalese et al. \cite{cicalese2011weightedtree}, who provide an $O(\log \NumBV/ \log\log\log \NumBV)$-approximation.
The sublogarithmic terms were improved by Cicalese et al.~in \cite{cicalese2014weightedtree}, and were then brought down to an $O(\sqrt{\log n})$-approximation and a QPTAS by Dereniowski et al.~\cite{dereniowski2017weightedtree}.

In our derivation of a PTAS on trees, we draw crucially on $k$-cut trees and their analysis.
They grew out of a line of research aiming to generalize optimality properties of dynamic binary search trees (BSTs) to the setting where the underlying search space is a tree rather than a path.
The goal here is to design dynamic search trees on trees whose runtime over a sequence of searches asymptotically approximates the runtime of the optimal dynamic search tree in the offline setting. The property of achieving a constant competitive ratio is called \textit{dynamic optimality}; a goal that has attracted much interest, but has yet to be achieved even for BSTs.
That line of work has led to the development of various data structures, such as Splay Trees \cite{sleator:tarjan:amortized}, the first to achieve a competitive ratio of $O(\log n)$, and Tango Trees \cite{demaine2007tango}, which were the first to achieve the best known competitive ratio of $O(\log \log n)$.
Heeringa, Iordan, and Theran \cite{heeringa2011stt} were the first to extend this problem to tree-like partial orders; they achieve a competitive ratio of $O(w)$ where $w$ is the width of the partial order.
This was then taken to the tree setting by Bose et al.~in \cite{bose2019stt}: they adapt ideas used in Tango Trees to yield the same competitive ratio of $O(\log \log n)$.
In particular, they introduced the notion of \emph{Steiner-closed} search trees, which are search trees that are able to extend certain nice properties of paths that BST's exploit to search trees on trees (STTs).
This notion was then generalized to \emph{$k$-cut} search trees by Berendsohn and Kozma \cite{berendsohn2020stt}; doing so allowed them to extend Splay Trees to the setting of STTs.

\section{Preliminaries}\label{s:prelims}
We use $\NN$ to denote the non-negative integers. 
  
We denote the underlying (simple, undirected) tree by $\BT = (\BV, \BE)$, with $\NumBV = \SetCard{\BV}$ vertices.
For a node $v \in \BV$, we use  $T \setminus \{v\}$ to denote the forest obtained by removing the vertex $v$ from $T$.
The \emph{distance} between $u$ and $v$ in \BT is denoted by \Dist[\BT]{u}{v}.
For a subset $S \subseteq \BV$, we write $\Boundary[\BT]{S} = \Set{u \in V \setminus S}{(u, v) \in \BE \text{ for some } v \in S}$ for the \emph{boundary} of $S$ in \BV.
We sometimes omit \BT from notation when it is clear from the context.

When \BT is specifically the line, i.e., we are considering the ``classical'' search setting, we identify $\BV = \SET{1, \ldots, \NumBV}$, with the understanding that $\BE = \Set{(j,j+1)}{j = 1, \ldots, \NumBV-1}$, and $\Dist{u}{v} = \Abs{u-v}$.

\subsection{The Query Model}
We now formally define the search model. We are interested in deterministic algorithms which uniquely identify a target vertex $\TV \in \BT$ chosen adversarially, by explicitly querying the target (even if it can be inferred from the prior queries).
The algorithm knows \BT a priori.
In each round $i$, the algorithm queries a node $\Query{i} \in \BT$.
In response, it learns either that $\Query{i} = \TV$ or which of the neighbors of \Query{i} is closer to \TV than \Query{i}, i.e., which component of $T \setminus \{\Query{i}\}$ contains \TV.
The process terminates when the algorithm queries \TV.
At any point, the set of nodes consistent with all answers to queries is called the \emph{feasible} set.

The cost of querying \Query{i} is \GenCost[\BT]{\Query{i}}{\TV}, with $\GenCost[\BT]{\TV}{\TV} = 0$, i.e., there is no cost for a correct query.
For most of the results on general trees \BT, we assume that this cost only depends on the distance between \Query{i} and \TV, i.e., it is of the form $\GenCost[\BT]{\Query{i}}{\TV} = \Cost{\Dist[\BT]{\Query{i}}{\TV}}$, where \COST is a non-decreasing non-negative function.
When \BT is specifically the line, there is natural meaning to ``guessing too low'' vs.~``guessing too high,'' so we consider a more general model where the cost function can be different depending on whether the guess is too low or too high: $\GenCost{\Query{i}}{\TV} = \Cost{\Query{i} - \TV}$. In this case, we assume that \COST is non-increasing for negative arguments, non-decreasing for positive arguments, and non-negative; in fact, for Theorem~\ref{t:multivar_path}, we consider a further generalization, where \GenCost{\Query{i}}{\TV} can be any monotone bounded-degree bivariate polynomial of \Query{i} and \TV with bounded coefficients.

The cost of a set \SEQ of queries, given a target \TV, is simply the sum of costs $\SeqCost{\SEQ} = \sum_{\QUERY \in \SEQ} \GenCost{\QUERY}{\TV}$.
We are interested only in deterministic algorithms \ALG which never query a vertex $v$ for which it can be inferred from previous answers that $v \neq \TV$;
we call such algorithms \emph{canonical}.
When searching on paths/trees, there is always an optimal canonical algorithm%
\footnote{Consider a query to a vertex \Query{i} which is known not to be the target vertex. Then, there must have been an earlier query \Query{j} with $j < i$ ruling out the possibility that $\Query{i} = \TV$. The response to \Query{j} must have revealed that $\TV \in S$, where $S$ is a particular subtree of \Query{j}, not containing \Query{i}. The response to \Query{i} will only reveal that \TV lies in the subtree $S'$ of \Query{i} containing \Query{j}, and $S' \supseteq S$. (Equality is possible if $\Query{i} = \Query{j}$.) Thus, the response to \Query{i} reveals no new information to the algorithm.};
notice that this is not true for a generalization to arbitrary graphs.
Canonical algorithms \ALG always terminate in at most \NumBV queries, and because they are deterministic, the query sequence $\langle \Query{1}, \Query{2}, \ldots \rangle =: \Seq{\ALG}{\TV}$ is determined completely by the target vertex.
The worst-case cost incurred by the algorithm for a given tree and cost function (omitted from the notation) is then defined as follows:
\begin{align*}
  \AlgCost{\ALG} & = \max_{\TV \in \BV} \SeqCost{\Seq{\ALG}{\TV}}.
\end{align*}
The goal is to design efficient algorithms \ALG to (exactly or approximately) minimize \AlgCost{\ALG}.

\subsection{An Alternative View: Search Trees}

An adaptive algorithm can naturally be viewed as a search tree, encoding the algorithm's decisions.
We formally define search trees as follows:

\begin{definition}[Search Tree on Trees] \label{def:search-trees}
  A \emph{search tree} on a tree (STT) \BT is a rooted and labeled tree \ST with the following property:
  The labels \Lab{u} of the nodes $u \in \ST$ are nodes of \BT, and 
  if the root $r$ of \ST is labeled \Lab{r}, each rooted subtree of $\ST \setminus \SET{r}$ is an STT for a different connected component of $\BT \setminus \SET{\Lab{r}}$; furthermore, the degree of $r$ in \ST equals the degree of \Lab{r} in \BT.
\end{definition}

Notice that this definition implies that each leaf node $u$ of \ST must be the STT for a singleton node; if not, then whichever node is chosen as \Lab{u} would have degree at least 1, implying that $u$ would have to have at least one child.
Furthermore, the labeling \LAB defines a bijection between the nodes of \ST and those of \BT.
\begin{definition} \label{def:feasible-node-set}
  Let \ST be an STT.
  For every node $u$ of \ST, we define $$\Feas[\ST]{u} = \Set{\Lab{x}}{x \text{ is a descendant of } u \text{ or } x=u}.$$
  We call \Feas[\ST]{u} the \emph{feasible set} associated with node $u$ (for reasons that will become obvious momentarily).
\end{definition}
There is a natural correspondence between STTs on \BT and deterministic search strategies \ALG on \BT.
\begin{enumerate}
\item For every STT \ST with root $r$ labeled \Lab{r}, the corresponding search strategy first queries \Lab{r}.
If the query response confirms that \Lab{r} is the target, the search terminates.
Otherwise, the response reveals the connected component $S$ of $\BT \setminus \SET{\Lab{r}}$ containing \TV, and the algorithm continues by recursively applying the transformation to the corresponding subtree of \ST.
Induction on the number of queries in the sequence easily shows that the set of feasible nodes at any point in the search (associated with a node $u$) equals \Feas[\ST]{u}.
  
\item Conversely, consider any canonical algorithm \ALG.
Its first query \Query{1} will be the vertex $\Lab{r} = \Query{1}$ of the root $r$ of the STT.
For each answer other than $\Query{1} = \TV$, the response reveals the connected component $S$ of $\BT \setminus \SET{\Query{1}}$ containing \TV.
Because the algorithm is canonical, it will not query any vertex outside of $S$ in future queries.
Thus, after the first query, the algorithm will be canonical for the component $S$.
By induction, the remaining execution of the algorithm has a corresponding STT.
Then, \ST is obtained by having each of the subtrees corresponding to components $S$ of $\BT \setminus \SET{\Query{1}}$ as the subtrees of $r$.
Again, induction easily shows that the set of feasible nodes at any point in the search (associated with a node $u$) equals \Feas[\ST]{u}.
\end{enumerate}

For any vertex $u \in \ST$, we write \Path[\ST]{u} for the set of vertices on the root-to-$u$ path in \ST, including $r$ and $u$.
Then, for any target node \TV, the cost of the search strategy corresponding to \ST can be written as
\begin{align*}
  \TreeCost{\ST}{\TV} & = \sum_{u \in \Path[\ST]{t}} \GenCost{\Lab{u}}{\TV}, 
\end{align*}
the total cost of queries performed according to the labels of the nodes in the search tree on the unique path all of whose vertices have \TV in their feasible node set.

\section{An Approximation Guarantee for Standard Binary Search}\label{s:binary}
Here, we show that the standard Binary Search Algorithm, which always queries the midpoint of the remaining interval, is a 4-approximation algorithm for any symmetric monotone cost function \COST, proving Theorem~\ref{t:binary-search}. 
The factor of 4 in the analysis is likely not tight. 
However, as we show in \cref{ss:approxlb}, even when $\Cost{x} = x$ is linear, Binary Search is no better than a $4/(\sqrt{33} - 3) \approx 1.457$ approximation, and for general cost functions, it is no better than a 2-approximation.

We denote the Binary Search algorithm by \BS, and begin with the following straightforward upper bound on the total cost incurred by $\BS$:
  $$\AlgCost{\BS} \leq \sum_{i=1}^{\Floor{\log_2(\NumBV)}} \Cost{\Floor{\NumBV/2^i}}.$$

To prove this bound, note that the median of a path of length $\ell$ is at distance at most $\Floor{\ell/2}$ from all points of the path.
When Binary Search makes its \Kth{i} query, the length of the remaining subpath is at most $\Floor{\NumBV/2^{i-1}}$.
Thus, the cost of the \Kth{i} query is at most $\Cost{\Floor{\Floor{\NumBV/2^{i-1}}/2}} \leq \Cost{\Floor{\NumBV/2^i}}$.
When the remaining subpath has length 1 (or 0), the remaining cost is 0.
Thus, the total cost is incurred in rounds $i$ with $\Floor{\NumBV/2^{i-1}} \geq 2$, which are $i \leq \log_2 \NumBV$.

The key part of the proof is establishing lower bounds on the cost of an arbitrary --- though without loss of generality canonical --- algorithm \ALG.
We will show the following three bounds:
\begin{align}
  \AlgCost{\ALG} & \geq \Cost{\Floor{\NumBV/2}}  \label{eqn:lower-bound-half} \\
  \AlgCost{\ALG} & \geq \Cost{\Floor{\NumBV/4}} + \Cost{\Floor{\NumBV/8}}  \label{eqn:lower-bound-quarter} \\
  2 \AlgCost{\ALG} & \geq \sum_{i=4}^{\Floor{\log_2(\NumBV)}} \Cost{\Floor{\NumBV/2^i}}.  \label{eqn:lower-bound-sixteenth}
\end{align}

Using these bounds, the approximation guarantee of 4 then follows immediately:
\begin{align*}
  4 \AlgCost{\ALG} & \geq \Cost{\Floor{\NumBV/2}} + (\Cost{\Floor{\NumBV/4}} + \Cost{\Floor{\NumBV/8}})
                     + \sum_{i=4}^{\Floor{\log_2(\NumBV)}} \Cost{\Floor{\NumBV/2^i}}\\
  \; &= \; \sum_{i=1}^{\Floor{\log_2(\NumBV)}} \Cost{\Floor{\NumBV/2^i}}
  \; \geq \; \AlgCost{\BS}.
\end{align*}

The bound from \cref{eqn:lower-bound-half} is straightforward: in response to a query of \Query{1}, the adversary chooses $\TV = 1$ or $\TV = \NumBV$, whichever is further from \Query{1}.
This node has distance at least $\Floor{\NumBV/2}$ from \Query{1}, so the first step alone incurs cost $\Cost{\Floor{\NumBV/2}}$. 

The bound from \cref{eqn:lower-bound-quarter} is also fairly direct.
In response to a query of \Query{1} (w.l.o.g.~$\Query{1} \geq 1 + \Floor{\NumBV/2}$), the adversary points in the direction of the larger remaining set; here, to the left.
If the next query is $\Query{2} \geq \Floor{\NumBV/8}+1$, the target is revealed to be at the left endpoint, and the total cost is at least $\Cost{\Query{1}-1} + \Cost{\Query{2}-1} \geq \Cost{\Floor{\NumBV/2}} + \Cost{\Floor{\NumBV/8}} \geq \Cost{\Floor{\NumBV/4}} + \Cost{\Floor{\NumBV/8}}$.
Otherwise (when $\Query{2} \leq \Floor{\NumBV/8}$), the target is revealed to be $\Ceiling{\NumBV/4}$. As a result, the total cost is at least $\Cost{1+\Floor{\NumBV/2}-\Ceiling{\NumBV/4}} + \Cost{\Ceiling{\NumBV/4} - \Floor{\NumBV/8}} \geq \Cost{\Floor{\NumBV/4}} + \Cost{\Floor{\NumBV/8}}$.

The bound from \cref{eqn:lower-bound-sixteenth} is significantly more involved.
We will analyze an adversary which in response to any query $q$ always points in the direction of the larger remaining subpath.
We remark that such an adversary may not be optimal (i.e., worst-case) when the cost function is not constant; for instance, it is typically better for the adversary to point away from more of the prior queries, even if the resulting subpath is shorter.
However, even this non-worst-case adversary is sufficient to establish the claimed lower bound.

Let \TV be the target revealed at the end.
Let $j=0, \ldots, \ell$ be the iterations of \ALG, and for each $j$,
let $L_j, R_j$ be such that the feasible set is $[L_j+1, R_j-1]$ after $j$ queries have been made; that is, $L_j, R_j$ are the infeasible nodes immediately to the left and right of the feasible set.
The fact that the feasible interval keeps shrinking implies that
\[
  0 = L_0 \leq L_1 \leq \cdots \leq L_{\ell} = \TV-1 < \TV + 1 = R_{\ell} \leq R_{\ell-1} \leq \cdots \leq R_1 \leq R_0 = \NumBV+1.
\]
For each $j$, let $z_j$ be the one of $L_j$ and $R_j$ maximizing the distance to \TV (the choice is arbitrary if they have the same distance), and let $d_j = \Dist{\TV}{z_j} = \Abs{\TV-z_j}$.

First, we claim that $j(0) = 1$ satisfies that $\frac{\NumBV}{4 \cdot 3^0} < d_{j(0)}$.
This holds because the adversary's choice of always pointing in the direction of the larger possible remaining feasible set ensures that $R_1 - L_1 > \NumBV/2$, so at least one of $L_1, R_1$ is at distance more than $\NumBV/4$ from \TV.

Next, we claim that for each $i=1, \ldots, 1 + \Floor{\log_3 (\NumBV/4)}$, there exists\footnote{There may be multiple such $j$, in which case $j(i)$ is an arbitrary one of them.} a $j=j(i) \geq 1$ with $\frac{\NumBV}{4 \cdot 3^i} < d_j \leq \frac{\NumBV}{4 \cdot 3^{i-1}}$.
Assume for contradiction that there exists no such $j$.
Let $j$ be largest such that $d_j > \frac{\NumBV}{4 \cdot 3^{i-1}}$; such a $j$ must exist, because again $j=1$ is an option.
Also, note that the upper bound on $i$ ensures that in step $j$, the feasible set contains more than one node.
By symmetry, assume without loss of generality that $z_j = R_j$.
If $|R_{j+1} - \TV| > \frac{\NumBV}{4 \cdot 3^{i-1}}$ or $|L_{j+1} - \TV| > \frac{\NumBV}{4 \cdot 3^{i-1}}$, this would contradict the maximality of $j$.
Hence, by the choice of $i$ (ruling out the intermediate range of distances), we obtain that $|R_{j+1} - \TV| \leq \frac{\NumBV}{4 \cdot 3^i}$ and $|L_{j+1} - \TV| \leq \frac{\NumBV}{4 \cdot 3^i}$.
This also implies that $R_{j+1} < R_j$, and thus $L_{j+1} = L_j$ by the update process.
As a result, query $j$ must have been at $R_{j+1}$, and the adversary pointed to the left, implying that $R_{j+1}-L_j \geq R_j-R_{j+1}$.
By the triangle inequality (applied with \TV in the middle), $R_{j+1} - L_j \leq \frac{2 \NumBV}{4 \cdot 3^i}$, so $R_j - R_{j+1} \leq \frac{2 \NumBV}{4 \cdot 3^i}$ as well.
Again by the triangle inequality (with $R_{j+1}$ in the middle), $R_j - \TV \leq \frac{\NumBV}{4 \cdot 3^i} + \frac{2 \NumBV}{4 \cdot 3^i} = \frac{\NumBV}{4 \cdot 3^{i-1}}$, contradicting that $d_j > \frac{\NumBV}{4 \cdot 3^{i-1}}$.

In summary, so far, we have shown that for each $i=0, \ldots, 1 + \Floor{\log_3 (\NumBV/4)}$, there exists a point $y_i = z_{j(i)}$ with $j(i) \geq 1$ that was immediately to the left or right of the feasible interval at some point during the execution of \ALG, and was at distance $d_{j(i)} \in (\frac{\NumBV}{4 \cdot 3^i}, \frac{\NumBV}{4 \cdot 3^{i-1}}]$ from \TV (except for $i=0$, where we did not establish an upper bound).
Because the corresponding distance intervals are disjoint, all of the points $y_i$ are distinct.
Furthermore, whenever $y_i \neq 0$ and $y_i \neq \NumBV + 1$, the fact that $y_i$ was immediately to the left/right of the feasible region must have been the result of querying $y_i$ during the algorithm.
Hence, for each $i$ with $y_i \notin \SET{0, \NumBV+1}$, \ALG must have incurred cost at least $\Cost{\Ceiling{\frac{\NumBV}{4 \cdot 3^i}}}$ with some query, and these queries are distinct for different $i$.
Because one of $\SET{0, \NumBV+1}$ was only an endpoint for iteration 0, at most \emph{one} (not two) of the $y_i$ can be in $\SET{0, \NumBV+1}$, so we can lower-bound the total cost of \ALG as follows:
\begin{align}
  \AlgCost{\ALG} & \geq \sum_{i=0, y_i \notin \SET{0, \NumBV}}^{1+\Floor{\log_3 (\NumBV/4)}} \Cost{\Ceiling{\frac{\NumBV}{4 \cdot 3^i}}}
                 \; \geq \; \sum_{i=1}^{1+\Floor{\log_3 (\NumBV/4)}} \Cost{\Ceiling{\frac{\NumBV}{4 \cdot 3^i}}} \label{eqn:power-of-three-bound}
\end{align}
by monotonicity of \COST.

Finally, we show that 
\begin{align*}
  \sum_{i=1}^{1+\Floor{\log_3 (\NumBV/4)}} \Cost{\Ceiling{\frac{\NumBV}{4 \cdot 3^i}}}
  & \geq \sum_{i=1}^{\Floor{\log_4(\NumBV)}-1} \Cost{\Floor{\frac{\NumBV}{4^{i+1}}}} 
\\  \sum_{i=1}^{1+\Floor{\log_3 (\NumBV/4)}} \Cost{\Ceiling{\frac{\NumBV}{4 \cdot 3^i}}}
  & \geq \sum_{i=1}^{\Floor{\log_4(\NumBV/2)}-1} \Cost{\Floor{\frac{\NumBV}{2 \cdot 4^{i+1}}}};
\end{align*}
then, adding these two inequalities, observing that
$$\sum_{i=1}^{\Floor{\log_4(\NumBV)}-1} \Cost{\Floor{\frac{\NumBV}{4^{i+1}}}} +
\sum_{i=1}^{\Floor{\log_4(\NumBV/2)}-1} \Cost{\Floor{\frac{\NumBV}{2 \cdot 4^{i+1}}}}
= \sum_{i=4}^{\Floor{\log_2(\NumBV)}} \Cost{\Floor{\frac{\NumBV}{2^i}}},$$
and combining the left-hand sides with \cref{eqn:power-of-three-bound} completes the proof.
To see the two inequalities, observe that the sums on the left-hand sides have more terms, and for any given $i$, we have that
$\Ceiling{\frac{\NumBV}{4 \cdot 3^i}} \geq \frac{\NumBV}{4 \cdot 3^i} \geq
\frac{\NumBV}{4^{i+1}} \geq \Floor{\frac{\NumBV}{4^{i+1}}}$, and similarly for the second sum (where the terms under the \COST function are smaller by another factor of 2).
The result now follows due to the monotonicity of \COST.

\subsection{Lower Bounds on Approximation Ratio}
\label{ss:approxlb}
We show that for general (symmetric) monotone cost functions, Binary Search is no better than a 2-approximation, and even when the cost function is linear, the approximation ratio is no better than $4/(\sqrt{33} - 3) \approx 1.457$.

To prove the lower bound of 2, consider the cost function $h(x)$ which is 1 if $x \geq \Floor{n/4}$ and 0 otherwise. To avoid rounding in the presentation\footnote{The result holds for arbitrary $n$.}, assume that $n=2^k-1$ for some $k$.
If the target is at 1, because the first two queries are at $\Ceiling{n/2}, \Ceiling{n/4}$, the cost of Binary Search is 2.
On the other hand, consider the following (optimal) strategy.
First query at $\Ceiling{\frac{n}{2}}$; w.l.o.g., assume that the target is less than $\Ceiling{\frac{n}{2}}$. Then, perform a query at $\Ceiling{\frac{n}{6}}$. If the target is less than $\Ceiling{\frac{n}{6}}$, the query is free, and so is a subsequent linear search over the remaining interval.
If the target is greater than $\Ceiling{\frac{n}{6}}$, then one performs a query at $\Ceiling{\frac{n}{3}}$. The combined cost of the queries at $\Ceiling{\frac{n}{6}}, \Ceiling{\frac{n}{3}}$, and $\Ceiling{\frac{n}{2}}$ is at most 1, since the adversary cannot choose a target that is at least a distance $\Floor{\frac{n}{4}}$ from more than one of the three queries. At this point, any remaining queries results in a cost of 0. Therefore, the cost of this strategy is 1, establishing a lower bound of 2 on the approximation ratio of Binary Search.

For linear cost functions $h(x) = x$, we consider the search on the continuous interval $[0,1]$ with the same search and cost model.
This is solely to avoid rounding in the presentation --- it is easy to verify that as $n \to \infty$, the lower bound converges to the one for the continuous model.
First observe that if the target is at 0, then Binary Search incurs cost $\sum_{i=1}^\infty 2^{-i} = 1$.
Now consider the following strategy (which is not optimal): first query at $\frac{1}{2}$; w.l.o.g., the target is less than $\frac{1}{2}$.
Then, query at $\gamma = \frac{\sqrt{33} - 5}{8}$. If the target is less than $\gamma$, perform binary search on the interval $[0, \gamma]$.
If the target is greater than $\gamma$, recurse on $[\gamma, \frac{1}{2}]$.
Let $C$ denote the cost of this search strategy.
In the first case, the cost is at most $\half + \gamma$ for the first two queries, plus $\gamma$ for the subsequent Binary Search, for a total of $\half + 2\gamma$.
In the second case, the cost is exactly $\half-\gamma$ for the first two queries combined (regardless where the target actually is within the interval), plus $C \cdot (\half-\gamma)$ for the recursive search, which is performed on an interval of size $\half-\gamma$; here, we used linearity of the cost function.
Hence, we have derived that 
\begin{align*}
    C\leq \max\left(\frac{1}{2} + 2 \gamma, \left(\frac{1}{2} - \gamma\right) \cdot (C + 1)\right).
\end{align*}

Substituting $\gamma = \frac{\sqrt{33}-5}{8}$ and $C = \frac{\sqrt{33}-3}{4}$, we can see that both terms are exactly equal to $C$.
Hence, $C$ is an upper bound on the cost of this strategy, and we obtain a lower bound of $1/C = \frac{4}{\sqrt{33}-3} \approx 1.457$ on the approximation ratio of Binary Search for linear cost functions.

\section{An Exact Algorithm for Paths}\label{s:paths}
We assume that the cost function $\Cost{x}$ is the form $\Cost{x} = \Cost[-]{-x}$ if $x < 0$ and $\Cost{x} = \Cost[+]{x}$ if $x > 0$, where
\begin{align*}
  \Cost[-]{x} & = \sum_{m=0}^{\DB} \PolCo[-]{m} x &
  \Cost[+]{x} & = \sum_{m=0}^{\DB} \PolCo[+]{m} x^k
\end{align*}
are polynomials of degree (at most) \DB.

To derive an optimal algorithm, consider the set of past queries that were too low (denoted by \SEQ[-]) or too high (denoted by \SEQ[+]).
For each query made by an algorithm, the response reveals which of the two sets the query belongs to.
Write $L = \max \SEQ[-]$ and $R = \min \SEQ[+]$; thus, the feasible set of targets at this point is the interval $\SET{L+1, L+2, \ldots, R-1}$.

For any set \SEQ of queries and $j \in \SET{0, \ldots, \DB}$, write $$\CumCo[\SEQ]{j}  = \sum_{q \in \SEQ} q^j.$$                

We can think of the \CumCo[\SEQ]{j} as \emph{sketches} of the sequences of prior queries. While there may be exponentially many possible sequences of queries, the sketches compress all relevant information so that only polynomially many sketches can arise. This is done by sketching the query sequence as coefficients of a re-parameterized cost polynomial.

When the target is $\TV \geq L$ (which it is guaranteed to be), the total cost for the queries in \SEQ[-] can be written as
\begin{align*}
  \SeqCost{\SEQ[-],\TV}
     & = \sum_{\QUERY \in \SEQ[-]} \Cost[-]{\TV - \QUERY} 
     \; = \;  \sum_{\QUERY \in \SEQ[-]} \sum_{m=0}^{\DB} \PolCo[-]{m} (\TV - \QUERY)^m\\
     \; & = \;  \sum_{\QUERY \in \SEQ[-]} \sum_{m=0}^{\DB} \PolCo[-]{m} \cdot \sum_{j=0}^m \binom{m}{j} \cdot (-1)^{m-j} \TV^j \QUERY^{m-j}
  \\ & =   \sum_{j=0}^{\DB} \TV^j \cdot \sum_{m=j}^{\DB} \PolCo[-]{m} \cdot \binom{m}{j} \cdot (-1)^{m-j} \cdot \CumCo[{\SEQ[-]}]{m-j}.
\end{align*}

By a similar calculation, since the target $\TV \leq R$, the total cost for the queries in \SEQ[+] can be written as
\begin{align*}
  \SeqCost{\SEQ[+],\TV}
  & =   \sum_{j=0}^{\DB} \TV^j \cdot \sum_{m=j}^{\DB} \PolCo[+]{m} \cdot \binom{m}{j} \cdot (-1)^{j} \cdot \CumCo[{\SEQ[+]}]{m-j}.
\end{align*}

First, given the \CumCo[{\SEQ[-]}]{j} and \CumCo[{\SEQ[+]}]{j}, the total cost of a target \TV can be directly calculated as $\SeqCost{\SEQ[-],\TV} + \SeqCost{\SEQ[+],\TV}$.
Second, notice that the \CumCo[\SEQ]{j} are linear, in the sense that if $\SEQ \cap \SEQ['] = \emptyset$, then $\CumCo[{\SEQ \cup \SEQ[']}]{j} = \CumCo[\SEQ]{j} + \CumCo[{\SEQ[']}]{j}$ for all $j$. This fact is very useful in formulating the recurrence relation.
Third, because \SEQ can contain at most \NumBV queries, each of which is bounded by \NumBV, all of the \CumCo[\SEQ]{j} are always integers and bounded by $\NumBV^{\DB+1}$.

Together, these observations enable an approach based on dynamic programming.
Subproblems are characterized by the left and right endpoint $L, R$ of the interval of remaining candidates as well as the coefficients $(\CumCo[-]{j})_{j=0}^{\DB}, (\CumCo[+]{j})_{j=0}^{\DB}$ of the cost polynomial.
Importantly, because the coefficients fully determine the cost that a target \TV will cause, the past queries \SEQ are sufficiently encoded in the coefficients \CumCo[-]{j}, \CumCo[+]{j}, and need not be explicitly kept track of.

Consider an input instance $\langle L, R, (\CumCo[-]{j})_{j=0}^{\DB}, (\CumCo[+]{j})_{j=0}^{\DB} \rangle$.
The optimum solution queries some $\QUERY \in \SET{L+1, \ldots, R-1}$.
In response, the adversary determines whether $\TV = \QUERY$, $\TV < \QUERY$, or $\TV > \QUERY$.
(If $\QUERY = L+1$ or $\QUERY = R-1$, the corresponding outcomes $\TV < \QUERY$ and $\TV > \QUERY$ are impossible, a fact we will encode by stating that the cost for the empty set is $-\infty$.)
In these three cases, the outcomes are as follows:
\begin{itemize}
\item If $\TV = \QUERY$, then the cost is obtained from the equations derived earlier and equals
\begin{align}
  \SeqCost{\langle (\CumCo[-]{j})_{j=0}^{\DB}, (\CumCo[+]{j})_{j=0}^{\DB} \rangle,\QUERY}
   = \sum_{j=0}^{\DB} \QUERY^j \cdot \sum_{m=j}^{\DB} \binom{m}{j} \cdot
    \left( \PolCo[-]{m} \cdot (-1)^{m-j} \cdot \CumCo[-]{m-j} + \PolCo[+]{m} \cdot (-1)^{j} \cdot \CumCo[+]{m-j} \right). \label{eqn:dist-poly-cost}
\end{align}
\item If $\TV < \QUERY$, then the target is known to lie in the interval $\SET{L+1, \ldots, \QUERY-1}$, and \QUERY was too high. Thus, for the remaining subproblem (with $R$ updated to \QUERY), the cost function is characterized by the same coefficients \CumCo[-]{j}, whereas the \CumCo[+]{j} are updated to $\CumCo[+]{j} + \CumCo[\SET{\QUERY}]{j}$. This subproblem is solved optimally.
\item If $\TV > \QUERY$, then similarly, the target is known to lie in the interval $\SET{\QUERY+1, \ldots, R-1}$, and \QUERY was too low. Thus, for the remaining subproblem (with $L$ updated to \QUERY), the cost function is characterized by the same coefficients \CumCo[+]{j}, whereas the \CumCo[-]{j} are updated to  $\CumCo[-]{j} + \CumCo[\SET{\QUERY}]{j}$. Again, this subproblem is solved optimally.
\end{itemize}
Thus, the optimal cost is characterized by the following recurrence for $R > L+1$ (with $\OPT(L, R, (\CumCo[-]{j})_{j=0}^{\DB}, (\CumCo[+]{j})_{j=0}^{\DB}) = -\infty$ when $R \leq L$):
\begin{align}
  \OPT(L, R, (\CumCo{j})_{j=0}^{\DB})
  =
    \min_{\QUERY = L+1, \ldots, R-1} \max \Big( &
       \SeqCost{\langle (\CumCo[-]{j})_{j=0}^{\DB}, (\CumCo[+]{j})_{j=0}^{\DB} \rangle,\QUERY}, \OPT(L, \QUERY, (\CumCo[-]{j})_{j=0}^{\DB}, (\CumCo[+]{j})_{j=0}^{\DB}) + (\CumCo[\SET{\QUERY}]{j})_{j=0}^{\DB}), \nonumber
  \\ &  \OPT(\QUERY, R, (\CumCo[-]{j})_{j=0}^{\DB} + (\CumCo[\SET{\QUERY}]{j})_{j=0}^{\DB}, (\CumCo[+]{j})_{j=0}^{\DB})
    \Big).\nonumber
\end{align}

This Dynamic Program can be solved bottom-up (from smallest to largest intervals $\SET{L, \ldots, R}$) in the standard way.
Because the coefficients $\CumCo[-]{j}, \CumCo[+]{j}$, being the sums of at most \NumBV \Kth{j} powers of integers bounded by \NumBV, are upper-bounded by $\NumBV^{j+1}$, encoding these $2(\DB+1)$ coefficients as well as the two endpoints $L, R$ requires a table of size $\NumBV^2 \cdot \prod_{j=0}^{\DB} \NumBV^{2(j+1)} = \NumBV^{2 + (\DB+1) \cdot (\DB+2)} = \NumBV^{\DB^2+3\DB+2}$; each entry can be computed in time $O(\NumBV)$, giving the claimed bound. 

This gives the desired polynomial-time algorithm, and proves \cref{t:path}.

\section{Position-Dependent Polynomial Costs}
\label{s:position-dependent}
Theorem~\ref{t:path} applies only to the case where the cost function \GenCost{\QUERY}{\TV} depends only on the \emph{distance} between the target and the query.
Thus, our results do not apply to the cost function for regret in the repeated posted-price and congestion control scenario, since there, the regret for $\QUERY > \TV$ is \TV, and thus dependent on \TV, but not through its distance from \QUERY.
Motivated by this, we study a more general form of cost functions, which are arbitrary bounded-degree polynomials of \QUERY and \TV.
The price we have to pay is to assume that all of the polynomial's \emph{coefficients} are bounded by a polynomial in \NumBV as well.
The more general form of the cost function is the following:

\begin{align*}
    \Cost[-]{\QUERY,\TV} & = \sum_{i, j \in \SET{0, \dots, \DB}: i + j \leq \DB} \PolCo[-]{i, j} \QUERY^i \TV^j
  & \Cost[+]{\QUERY,\TV} & = \sum_{i, j \in \SET{0, \dots, \DB}: i + j \leq \DB} \PolCo[+]{i, j} \QUERY^i \TV^j.
\end{align*}
All coefficients \PolCo{i,j} are integers and bounded by $O(n^s)$ for a constant $s$.
Similar to the analysis in \cref{s:paths}, the total cost for the queries $\SEQ[-], \SEQ[+]$ can be written as

\begin{align*}
  \SeqCost{\SEQ[-] \cup \SEQ[+],\TV}
     & = \sum_{\QUERY \in \SEQ[-]} \Cost[-]{\QUERY,\TV} + \sum_{\QUERY \in \SEQ[+]} \Cost[+]{\QUERY,\TV} \\
       \; & = \; \sum_{\QUERY \in \SEQ[-]} \sum_{i, j \in \{0, \dots, \DB\}: i + j \leq \DB} \PolCo[-]{i, j} \QUERY^i \TV^j
             + \sum_{\QUERY \in \SEQ[+]} \sum_{i, j \in \{0, \dots, \DB\}: i + j \leq \DB} \PolCo[+]{i, j} \QUERY^i \TV^j
  \\ &  = \sum_{j=0}^{\DB} \TV^j \cdot \left( \sum_{\QUERY \in \SEQ[-]} \sum_{k = 0}^{\DB - j} \PolCo[-]{j, k} \QUERY^k 
                                             + \sum_{\QUERY \in \SEQ[+]} \sum_{k = 0}^{\DB - j} \PolCo[+]{j, k} \QUERY^k \right).
\end{align*}
Writing
\begin{align*}
  \CumCoT[{\SEQ[-] \cup \SEQ[+]}]{j}  & := \sum_{\QUERY \in \SEQ[-]} \sum_{k = 0}^{\DB - j} \PolCo[-]{j, k} \QUERY^k
                                         + \sum_{\QUERY \in \SEQ[+]} \sum_{k = 0}^{\DB - j} \PolCo[+]{j, k} \QUERY^k,
\end{align*}
we see that for a fixed past query set $\SEQ[-] \cup \SEQ[+]$, the cost for a given target $t$ can be evaluated knowing only the $\CumCoT[{\SEQ[-] \cup \SEQ[+]}]{j}$, as
\begin{align}
  \SeqCost{\SEQ[-] \cup \SEQ[+],\TV} & = \sum_{j=0}^{\DB} \CumCoT[{\SEQ[-] \cup \SEQ[+]}]{j} \cdot \TV^j. \label{eqn:general-poly-cost}
\end{align}
The $\CumCoT[{\SEQ[-] \cup \SEQ[+]}]{j}$ are also linear in the sets in the same sense as before, so $\langle L, R, (\CumCoT{j})_{j=0}^{\DB} \rangle$ provide a sufficient state encoding to enable the same Dynamic Program to apply, with the cost of revealing the target \TV given by \cref{eqn:general-poly-cost} instead of \cref{eqn:dist-poly-cost}.

Because all of the $\PolCo[-]{k,j}, \PolCo[+]{k,j}$ are bounded by $O(\NumBV^{\CB})$, $\SEQ[-] \cup \SEQ[+]$ is a set of size at most \NumBV, and $\QUERY \leq \NumBV$, all of the \CumCoT{j} are bounded in absolute value by $O(\NumBV^{\CB + \DB -j + 1})$.
Thus, the table has size $n^2 \cdot \prod_{j=0}^{\DB} n^{\CB + \DB - j +1} = O(n^{2 + (p+1) \cdot (p + 2s + 2)/2}),$ and each entry is computed in time $O(\NumBV)$.
This proves Theorem~\ref{t:multivar_path}.

\section{A PTAS for Trees}\label{s:trees}
In extending the approach of \cref{s:paths} to general trees, we face the following obstacles:

\begin{enumerate}
  \item Any subtree of \BT could arise as a remaining feasible set. Trying to implement Dynamic Programming in the obvious way would thus force an algorithm to consider exponentially many subtrees (rather than $O(\NumBV^2)$ for the path).
  \item The construction of the coefficients of the polynomial to summarize the past queries relied on the fact that each query was to the left or right of the remaining feasible interval. Such a simple partition (with corresponding signs) is not possible for more general subtrees.
\end{enumerate}

In order to address the first challenge, we use the $k$-cut STT framework introduced by Berendsohn and Kozma \cite{berendsohn2020stt}.
A $k$-cut search tree is an STT such that every feasible set is a subtree separated from its complement by at most $k$ edges.
Thus, there are at most $O(n^k)$ such subtrees, and our algorithm will enumerate only over those subtrees.
The key is then to analyze the loss in the approximation guarantee resulting from this restriction.
To address the second challenge, we then develop a method of encoding the cost imposed by the prior queries in a way that allows combining the encodings of different sets when the queried vertex sets merge by merging their respective subtrees.

\subsection{$k$-cut Search Trees on Trees} \label{s:k-cut}

\begin{definition}[$k$-Cut Tree, Definition 5 of \cite{berendsohn2020stt}] \label{def:kSTT}
  For $k \geq 1$, an STT \ST on tree $\BT=(\BV,\BE)$ is a \emph{$k$-cut tree} if $\SetCard{\Boundary[\BT]{\Feas[\ST]{u}}} \leq k$ for every $u \in \ST$,
  i.e., if each feasible set of vertices that can occur under this search tree is separated from the eliminated vertices by at most $k$ edges (or, equivalently, it has at most $k$ adjacent vertices).
\end{definition}

We analyze the cost increase that is incurred due to focusing only on $k$-cut STTs.
Our main lemma is the following generalization of a result of \cite{berendsohn2020stt}:

\begin{lemma} \label{t:cost_inflation}
  Let $\BT=(\BV,\BE)$ be a tree, $k \geq 3$, $\COST: \NN \to \RR$ a monotone non-decreasing cost function, and $\GenCost{\QUERY}{\TV} = \Cost{\Dist{\QUERY}{\TV}}$ the cost associated with query \QUERY when the target is \TV.
  Let \ST be an STT on \BT.
  Then, there exists a $k$-cut STT \STP on \BT such that for all target vertices \TV:
  \begin{align*}
      \TreeCost{\STP}{\TV} & \leq \left( 1 + \frac{1}{\Ceiling{k/2} - 1} \right) \cdot \TreeCost{\ST}{\TV}.
  \end{align*}
  In particular, the cost associated with the search tree (the maximum cost over all targets) increases by at most a factor of $\left( 1 + \frac{1}{\Ceiling{k/2} - 1} \right)$.
\end{lemma}

To prove the lemma, we use a conversion procedure due to Berendsohn and Kozma \cite{berendsohn2020stt} (Algorithm 1 of \cite{berendsohn2020stt}) which converts a general STT to a $k$-cut STT.
Their goal was also to reduce the search space for adaptive query strategies, at the cost of a small loss in solution quality.
However, they were interested only in one objective: the total number of queries; this corresponds to the special case when $\Cost{\DIST} = 1$ for all $\DIST \geq 1$.
Then, the maximum search cost associated with an STT \ST is exactly the height of \ST.
In \cite{berendsohn2020stt}, Berendsohn and Kozma show that for every $k \geq 3$, every STT \ST on \BT can be converted into a $k$-cut STT \STP on \BT such that
\begin{align*}
    \text{height}(\STP) & \leq \left( 1 + \frac{1}{\Ceiling{k/2} - 1} \right) \cdot \text{height}(\ST).
\end{align*}

The essence of their construction is an iterative transformation of \ST.
Whenever the feasible set \Feas{u} of some vertex $u \in \ST$ has boundary size at least $k$ (i.e., right on the boundary of violating the $k$-cut STT property), the algorithm chooses the next query judiciously to be one that separates the leaves of the boundary of the feasible set so that in the next step, no feasible set contains more than $k/2$ leaves.
This operation is achieved by ``promoting'' the node corresponding to the required query, i.e., moving it up in in the search tree.
The promotion may increase the search path for other nodes, but by at most 1, and on any given path, such a promotion can only happen roughly every $k/2$ steps.
This gives their desired bound.

We adapt these main ideas in our proof of Lemma~\ref{t:cost_inflation}, to work with more general non-decreasing cost functions.
Specifically, we will show that, even for the more general query costs, the same construction results in an identical constant multiplicative loss, by charging the increased cost due to the extra query to a subset of the queries made previously.

\subsection{$k$-cut Search Trees on Trees} \label{s:k-cut-appendix}

We introduce the following key definitions and basic results from \cite{berendsohn2020stt}.

\begin{definition}[Convex Hull, Definition 9 of \cite{bose2019stt}] \label{def:convex-hull}
  Let $\BT = (\BV,\BE)$ be a tree. The \textit{convex hull} of a set of vertices $S \subseteq \BV$, denoted by \ch[\BT]{S}, is the subtree of \BT induced by the union of all paths between pairs of vertices in $S$, i.e., the smallest connected subset of \BT containing all of $S$.
\end{definition}

\begin{definition}[Leaf Centroid, Definition 6 from \cite{berendsohn2020stt}]
  Let $\BT = (\BV,\BE)$ be a tree on $\NumBV \geq 3$ nodes, with $\ell$ leaves.
  A non-leaf node $v$ of $\BT$ is called a \emph{leaf centroid} if every connected component of $\BT \setminus \SET{v}$ has at most $\Floor{\frac{\ell}{2}} + 1$ leaves.
\end{definition}

A leaf centroid is always guaranteed to exist and can be found in linear time by starting from an arbitrary vertex $v$, and repeatedly moving along an edge in the direction of a component that has too many leaves.

The following basic facts about the boundaries of feasible sets of STTs were shown in \cite{berendsohn2020stt}.

\begin{proposition}[Lemma 2 of \cite{berendsohn2020stt}] \label{p:stt-boundary}
Let \ST be an STT with root $r$ on a tree $\BT = (\BV,\BE)$. Then:
\begin{enumerate}
   \item $\Boundary[\BT]{\Feas{u}} \subseteq \Set{\Lab{u'}}{u' \in \Path[\ST]{u}}$.
   \item Let $u \in \ST, u \neq r$ and $p$ the parent of $u$ in \ST. Then, $\SetCard{\Boundary[\BT]{\Feas{u}}} \leq \SetCard{\Boundary[\BT]{\Feas{p}}} + 1$.
\end{enumerate}
\end{proposition}

The first part of Proposition~\ref{p:stt-boundary} states that the boundary of the feasible set can only contain vertices that have been queried previously.
The second part states that the boundary size of the feasible set can increase by at most 1 upon querying a vertex.

The main reason for defining the convex hull earlier is that the leaves of the convex hull of the boundary of a feasible set must be a subset of the past queries made.

\begin{proposition}
\label{p:convex-hull}
    Let $\BT = (\BV,\BE)$ and $S \subseteq V$ such that the induced subgraph $\BT[S]$ is connected. Then, the set of leaves of $\ch[\BT]{\Boundary[\BT]{S}}$ is equal to $\Boundary[\BT]{S}$.
\end{proposition}
\begin{proof}
  Notice that $S' := \ch[\BT]{\Boundary[\BT]{S}}$ must be a subtree of $S \cup \Boundary[\BT]{S}$.
  $S'$ also is the disjoint union of $S \cap S'$ and $\Boundary[\BT]{S}$.
  By definition of the convex hull, only vertices in $\Boundary[\BT]{S}$ can be leaves of $S'$, proving one inclusion direction.
  For the converse direction, because $\BT[S]$ is connected, no vertex in $\Boundary[\BT]{S}$ can be adjacent to more than one vertex of $S$ (otherwise, a cycle would be formed), so each vertex of $\Boundary[\BT]{S}$ is in fact a leaf of $S'$.
\end{proof}

\subsubsection{Rotations and Promotions in STTs}

A very useful analysis tool is the concept of rotations on an STT.
Such rotations, essentially identical to rotations on standard search trees, correspond to local changes in the query order of vertices, and their effects can be isolated and bounded.

\begin{definition}[Separatedness, Betweenness] \label{def:betweenness}
  Consider a tree $\BT = (\BV,\BE)$.
  \begin{enumerate}
  \item When the vertex $v$ is on the unique path from $u$ to $w$, we say that $v$ \emph{separates} $u$ and $w$ (or, $u$ from $w$, and vice versa).
  Equivalently, $v$ separates $u$ and $w$ if $u$ and $w$ are in distinct connected components of $\BT \setminus \SET{v}$.

  \item We say that a vertex $v \in \BV$ is \emph{between} vertices $u \neq w \in \BV$ if $v$ is in the connected component of $\BT \setminus \SET{u, w}$ which is adjacent to both $u$ and $w$.\footnote{Equivalently, $v$ is between $u$ and $w$ if there is a path from $v$ to $u$ that does not contain $w$, and a path from $v$ to $w$ that does not contain $u$. Another equivalent form is that $v$ lies in the same connected component of $\BT \setminus \SET{u, w}$ as some vertex separating $u$ and $w$.} We denote this relation with $v \in [u, w]$, writing $[u,w]$ for the set of all vertices between $u$ and $w$.
  \end{enumerate}
\end{definition}

\begin{definition}[Rotations on STTs] \label{def:rotations}
  Fix a tree $\BT = (\BV,\BE)$, and let \ST be an STT on \BT.
  Let $u \in \ST$ be a vertex of \ST, and $p$ its parent.
  The \emph{rotation} of $u$ towards $p$ in \ST results in the STT \STP with the following properties:
    \begin{enumerate}
        \item $u$ is the parent of $p$ in \STP.
        \item Every child $u'$ of $u$ with $\Lab{u'} \in [\Lab{u}, \Lab{p}]$ becomes a child of $p$ in \STP.
        \item Parent-child relations in \STP remain the same as in \ST for all other pairs of vertices.
    \end{enumerate}
\end{definition}

For the special case when \BT is a path, it is easy to check that this definition coincides with the usual notion of rotation in a search tree.
Rather than operating on STTs with single rotations, we will often be considering sequences of rotations that move a specific vertex $u$ to the position of one of its ancestors $x$.

\begin{definition}[Promotions on STTs] \label{def:promotions}
  Fix a tree $\BT = (\BV,\BE)$, and let \ST be an STT on \BT.
  Let $u \in \ST$ be a vertex of \ST, and $x$ an ancestor of $u$.
  Let $u_0 = u$, and $u_{i+1}$ the parent of $u_i$ in \ST, for all $i$, with $u_{\ell} = x$.
  The \emph{promotion} of $u$ to $x$ in \ST is obtained by rotating $u$ towards $u_1$, then towards $u_2$, and so on, and finally towards $u_{\ell} = x$.
\end{definition}

The effect of promotions on feasible sets is captured by the following proposition.

\begin{proposition}\label{p:promo_feasible_set}
  Let \BT be a tree, and \ST an STT on \BT.
  Let $u \in \ST$ be a (non-root) vertex, and $x$ an ancestor of $u$ in \ST.
  Let \STP be the STT obtained by promoting $u$ to $x$.
  Then, $\Feas[\STP]{u} = \Feas[\ST]{x}$.
\end{proposition}

\begin{proof}
  When $u$ is promoted to $x$, the set of nodes in the subtree of \STP rooted at $u$ is equal to the set of nodes in the subtree of \ST rooted at $x$.
  The claim now follows immediately from Definition~\ref{def:feasible-node-set}.
\end{proof}

We are particularly interested in the effect of promoting a leaf centroid, which is captured by the following proposition from \cite{berendsohn2020stt}, stating that promoting a leaf centroid essentially halves the size of the boundary of the feasible set:

\begin{proposition}[Lemma 5 of \cite{berendsohn2020stt}] \label{p:leaf-centroid-promotion}
  Let $u \in \ST$ be such that $\SetCard{\Boundary[\BT]{\Feas{u}}} \geq k$.
  Let $v \in \BV$ be a leaf centroid of \ch{\Boundary[\BT]{\Feas{u}}}, and $x \in \ST$ the vertex with $\Lab{x} = v$.
  Let \STP be the STT obtained from promoting $x$ to $u$ in \ST.
  Then, for every child $u'$ of $x$ in \STP, we have $\SetCard{\Boundary[\BT]{\Feas[\STP]{u'}}} \leq \Floor{\SetCard{\Boundary[\BT]{\Feas[\ST]{u}}}/2} + 1$.
\end{proposition}

\subsection{Proof of Lemma~\ref{t:cost_inflation}} \label{ss:cost_increase}

We first analyze how a single promotion affects the search tree cost.
  
\begin{proposition} \label{p:promotion_cost}
  Let \BT be a tree, and \ST an STT on \BT.
  Let $u \in \ST$ be a (non-root) vertex and $x$ an ancestor of $u$ in \ST.
  Let $u_0 = u$, and $u_{i+1}$ the parent of $u_i$ in \ST, for all $i$, with $u_{\ell} = x$.
  Let \STP be obtained by promoting $u$ to $x$.
  Then, the new costs associated with searching for a target \TV in \STP satisfy the following:
  \begin{itemize}
     \item For all target nodes $\TV \in \bigcup_{j=1}^{\ell} (\Feas[\STP]{u_j} \setminus [\Lab{u}, \Lab{u_j}])$, the cost increases by \GenCost{\Lab{u}}{\TV}:
          $\TreeCost{\STP}{\TV} = \TreeCost{\ST}{\TV} + \GenCost{\Lab{u}}{\TV}$.
     \item For all other target nodes, the cost stays the same or decreases:
         $\TreeCost{\STP}{\TV} \leq \TreeCost{\ST}{\TV}$.
  \end{itemize}
\end{proposition}

\begin{proof}
  Each node $u_1, \ldots, u_{\ell}$ has $u$ added in its root-to-node path \Path[\ST]{u_i}.
  The cost of the query to \Lab{u} will be incurred whenever the target \TV is in \Feas[\STP]{u_i}.
  For the nodes in $[\Lab{u}, \Lab{u_i}]$, the cost had already been incurred previously, as they had been in \Feas[\ST]{u}; for the other nodes, this causes the claimed cost increase.
  For all other target nodes \TV, the path \Path[\ST]{x} to the node $x$ with $\Lab{x} = \TV$ stays the same in \STP, so the query sequence incurs the same cost.
\end{proof}

Next, we analyze the effect of the promotion of one leaf centroid when the feasible set associated with a node has large boundary.

\begin{proposition}\label{p:ancestor_sep}
  Let $\BT = (\BV, \BE)$ be a tree, and \ST an STT on \BT.
  Let $x \in \ST$ be such that $\SetCard{\Boundary[\BT]{\Feas[\ST]{x}}} = k$ for some $k \geq 3$.
  Let $v$ be a leaf centroid of \Feas[\ST]{x}, and $u \in \ST$ the descendant of $x$ in \ST such that $\Lab{u} = v$.
  Let \STP be the STT on \BT obtained from promoting $u$ to $x$.
  Then, for every target vertex $\TV \in \Feas[\STP]{u} \setminus \SET{v}$, we have that $v$ separates \TV from at least $\Ceiling{k/2}-1$ leaves of $\ch[\BT]{\Boundary[\BT]{\Feas[\STP]{u}}}$, or equivalently, from \Lab{u'} for at least $\Ceiling{k/2}-1$ ancestors $u'$ of $u$ in \STP.
\end{proposition}
\begin{proof}
  Let $S = \ch[\BT]{\Boundary[\BT]{\Feas[\STP]{u}}}$ be the convex hull of the boundary of $\Feas[\STP]{u}$ in $\STP$.
  Since promoting $u$ to $x$ in \ST results in the same feasible set, i.e. $\Feas[\STP]{u} = \Feas[\ST]{x}$, the number of leaves of $S$ is $k$.
  By Proposition~\ref{p:convex-hull}, the leaves of $S$ are all in \Boundary[\BT]{\Feas[\STP]{u}}; thus, by Proposition~\ref{p:stt-boundary}, the leaves must be the labels of ancestors of $u$ in \STP.
  
  By definition of leaf centroids, every connected component of $S \setminus \SET{v}$ contains at most $\Floor{k / 2} + 1$ leaves of $\BT[S]$.
  Thus, every target \TV in $S$ is separated by $v$ from at least $k - (\Floor{k / 2}+1) = \Ceiling{k / 2}-1$ leaves of $\BT[S]$.
\end{proof}

Proposition~\ref{p:ancestor_sep} states that if we have a feasible set $S$ separated from its complement in \BT by exactly $k$ edges, then querying a leaf centroid of $S$ instead of the query $x$ that \ST would have made separates at least $\Ceiling{\frac{k}{2}} - 1$ of the previously queried vertices from every vertex in the feasible set.
This enables a charging type argument at the heart of the proof of Lemma~\ref{t:cost_inflation}, as cost increases (according to Proposition~\ref{p:promotion_cost}) associated with promoting a leaf centroid can be averaged out and charged to at least $\Ceiling{k/2}-1$ of the prior queries.
This argument exploits the monotonicity of the cost function \COST.

\begin{extraproof}{Lemma~\ref{t:cost_inflation}}
    Fix $k \geq 3$. We describe an iterative procedure, introduced by Bose et al.~in \cite{bose2019stt}, and generalized to $k$-cut search trees in Algorithm 1 of \cite{berendsohn2020stt}, for explicitly obtaining the $k$-cut STT \STP from an STT \ST.

    In each iteration, the procedure promotes a leaf centroid of a node $u$ whose boundary is too large.
    Specifically, the procedure starts with $\STP = \ST$.
    Consider an iteration in which there exists at least one node $u \in \STP$ with $\SetCard{\Boundary[\BT]{\Feas[\STP]{u}}} \geq k$.
    Let $u$ be such a node with the property that all ancestors $x$ of $u$ have $\SetCard{\Boundary[\BT]{\Feas[\STP]{x}}} < k$.
    By Proposition \ref{p:stt-boundary}, $\SetCard{\Boundary[\BT]{\Feas[\STP]{u}}} = k$.
    Let $v \in \BV$ be a leaf centroid of $\ch{\Boundary[\BT]{\Feas[\STP]{u}}}$, and $x$ the descendant of $u$ with $\Lab{x} = v$.
    The procedure now promotes $x$ to $u$. Such a step is shown pictorially in \cref{f:leaf-centroid-promotion}.  

    \begin{figure}[htb]
    \begin{center}
    \begin{tikzpicture}
    \begin{scope}[every node/.style={circle,thick,draw}]
        \node (A) at (0, 0) {A};
        \node (B) at (1.414, 0) {B};
        \node (C) at (2.828, 0) {C};
        \node (D) at (0, 1.414) {D};
        \node (E) at (0, 2.828) {E};
        \node (F) at (-1.414, 0) {F};
        \node (G) at (-2.828, 0) {G};
        \node (H) at (0, -1.414) {H};
        \node (I) at (0, -2.828) {I};
        \node (J) at (0, -4.432) {J};
    \end{scope}
    
    \begin{scope}[>={Stealth[black]},
                  every node/.style={fill=white,circle},
                  every edge/.style={draw=black,very thick}]
        \path (A) edge (B);
        \path (B) edge (C);
        \path (A) edge (D);
        \path (D) edge (E);
        \path (A) edge (F);
        \path (F) edge (G);
        \path (A) edge (H);
        \path (H) edge (I);
        \path (I) edge (J);
    \end{scope}
    \end{tikzpicture}\vspace{1.25cm}\\
    \begin{tikzpicture}
    \begin{scope}[every node/.style={circle,thick,draw}]
        \node (B) at (0, 0) {B};
        \node (C) at (-1, -1) {C};
        \node (D) at (1, -1) {D};
        \node (E) at (0, -2) {E};
        \node (F) at (2, -2) {F};
        \node (G) at (1, -3) {G};
        \node (I) at (3, -3) {I};
        \node (J) at (2, -4) {J};
        \node (A) at (4, -4) {A};
        \node (H) at (4, -5.414) {H};
    \end{scope}
    
    \begin{scope}[>={Stealth[black]},
                  every node/.style={fill=white,circle},
                  every edge/.style={draw=black,very thick}]
        \path (B) edge (C);
        \path (B) edge (D);
        \path (D) edge (E);
        \path (D) edge (F);
        \path (F) edge (G);
        \path (F) edge (I);
        \path (I) edge (J);
        \path (I) edge (A);
        \path (A) edge (H);
        \path [dotted, ->] (A) edge [bend right = 70] (I);
    \end{scope}
    \end{tikzpicture}
    \begin{tikzpicture}
    \begin{scope}[every node/.style={circle,thick,draw}]
        \node (B) at (0, 0) {B};
        \node (C) at (-1, -1) {C};
        \node (D) at (1, -1) {D};
        \node (E) at (0, -2) {E};
        \node (F) at (2, -2) {F};
        \node (G) at (1, -3) {G};
        \node (A) at (3, -3) {A};
        \node (J) at (5, -5) {J};
        \node (I) at (4, -4) {I};
        \node (H) at (3, -5) {H};
    \end{scope}
    
    \begin{scope}[>={Stealth[black]},
                  every node/.style={fill=white,circle},
                  every edge/.style={draw=black,very thick}]
        \path (B) edge (C);
        \path (B) edge (D);
        \path (D) edge (E);
        \path (D) edge (F);
        \path (F) edge (G);
        \path (I) edge (J);
        \path (I) edge (H);
        \path (I) edge (A);
        \path (A) edge (F);
    \end{scope}
    \end{tikzpicture}
    \vspace{0.4cm}
    \caption{Leaf centroid promotion for converting to a 3-cut search tree. \\
      (Top): The base tree $\BT$.
      (Left): The original search tree $\ST$, showing the labels (queried nodes) for each node of \ST.
      The vertex with label $I$ has associated feasible set $\Feas{I} = \SET{A, H, I, J}$, which has boundary size 3.
      Although this is not immediately an issue, children of $I$ (such as $A$) might have feasible sets exceeding the boundary size.
      Thus, the procedure promotes a leaf centroid of $\ch{\Boundary{\Feas{I}}} = \ch{\SET{B, D, F}} = \SET{A, B, D, F, H, I, J}$ to $I$.
      Since $A$ is a leaf centroid (in fact, the only one in this example), the procedure promotes $A$ to $I$.
      Note that a leaf centroid is not necessarily a child of the vertex it is promoting to.
      The promotion only increases the cost of $J$ by $\Cost{\Dist[\BT]{A}{J}} = \Cost{3}$.
      In the proof of Theorem~\ref{t:cost_inflation}, we use the fact that $A$ separates $J$ from $B,D$, and $F$ --- which were already on the search path $\Path[\ST]{J}$ --- to charge the cost increase of \Cost{3} averaged over the cost of $\Ceiling{k/2} - 1= 1$ vertex from $\SET{B,D,F}$.
      Notice that all of $\SET{B,D,F}$ are further from $J$ than $A$ due to separation.
      (Right): The search tree \STP obtained from promoting $A$ to $I$.
      While previously, \Feas[\ST]{A} had boundary size 4, the promotion decreased the boundary size of \Feas[\STP]{A} to 1. \label{f:leaf-centroid-promotion}}
    \end{center}
    \end{figure} 

    We now analyze the effect of these promotions on the cost of searching for $\TV \in \STP$; fix a \TV for the rest of the proof.
    By Proposition~\ref{p:promotion_cost}, in the worst case, promoting $x$ to $u$ for $\TV \in \Feas[\ST]{u}$ increases the cost of searching for \TV by $\Cost{\Dist[\BT]{v}{\TV}}$.
    Proposition~\ref{p:ancestor_sep} guarantees that a subset of the leaves of $\ch{\Boundary[\BT]{\Feas[\ST]{u}}}$ of size at least $\Ceiling{k/2} - 1$ is separated from \TV by $v$.
    In other words, at least $\Ceiling{k/2} - 1$ of the nodes \Lab{u'} for ancestors $u'$ of $u$ in \STP are separated from \TV by $v$.
    We denote this subset by $S_v$, with $\SetCard{S_v} \geq \Ceiling{k/2} - 1$. 

    Because $v$ separates \TV from each $v' \in S_v$, $v$ lies on the path from \TV to $v'$, implying that $\Dist[\BT]{v'}{\TV} > \Dist[\BT]{v}{\TV}$.
    Since \COST is non-decreasing, the cost increase associated with the promotion of $x$ is upper-bounded by $\Cost{\Dist[\BT]{v'}{\TV}}$ for each $v' \in S_v$.
    Since this upper bound holds for all $v' \in S_v$, we can average the upper bound over $S_v$, implying that the cost increase of promoting $x$ to $u$ is bounded as follows:
    \begin{align*}
      \Cost{\Dist[\BT]{v}{\TV}} & \leq  \frac{1}{\SetCard{S_v}} \cdot \sum_{v' \in S_v} \Cost{\Dist[\BT]{v'}{\TV}}
                         \; \leq \; \frac{1}{\Ceiling{k/2} -1} \cdot \sum_{v' \in S_v} \Cost{\Dist[\BT]{v'}{\TV}}.
    \end{align*}

    Now, among all vertices of \ST that were promoted over the course of the procedure, let $x_1, x_2, \dots, x_{\ell}$ be the ones that ended up on the root-to-\TV path in \STP, in the order in which they were promoted.
    Let $v_i = \Lab{x_i}$ be the corresponding leaf centroids.
    We argue that, crucially, the corresponding sets $S_{v_1}, \ldots, S_{v_{\ell}}$ are disjoint.

    By the processing order in the procedure, whenever $i < j$, node $x_i$ is an ancestor of $x_j$.
    Because $\TV \in \Feas[\STP]{x_j} \subseteq \Feas[\STP]{x_i}$, all of \Feas[\STP]{x_j} lies in the same connected component of $\BT \setminus \SET{v_i}$.
    By definition of $S_{v_i}$, the target \TV is separated from $S_{v_i}$ by $v_i$; as a result, so is all of \Feas[\STP]{x_j}.
    Thus, all nodes of $S_{v_i}$ are at distance at least 2 from any node in \Feas[\STP]{x_j}; in particular, $S_{v_i} \cap \Boundary{\Feas[\STP]{x_j}} = \emptyset$.
    Because $S_{v_j} \subseteq \Boundary{\Feas[\STP]{x_j}}$, this implies the claimed disjointness.

    Further, by Proposition~\ref{p:stt-boundary}, the $S_{v_j}$ are all contained in $\Set{\Lab{u}}{u \in \Path[\ST]{\TV}}$.

    Applying these facts in the second inequality below, we can then bound the total cost increase for searching for \TV by
    \begin{align*}
      \TreeCost{\STP}{\TV} - \TreeCost{\ST}{\TV}
        & \leq \sum_{i = 1}^{\ell} \frac{1}{\Ceiling{k/2} - 1} \cdot \sum_{v' \in S_{v_i}} \Dist[\BT]{v'}{\TV} %\\
        \; = \; \frac{1}{\Ceiling{k/2}-1} \cdot \sum_{i=1}^{\ell} \sum_{v' \in S_{v_i}} \Dist[\BT]{v'}{\TV} \\
        & \leq \frac{1}{\Ceiling{k/2} - 1} \cdot \sum_{v' \in \Set{\Lab{u}}{u \in \Path[\ST]{\TV}}} \Dist[\BT]{v'}{\TV} %\\
        \; = \; \frac{1}{\Ceiling{k/2} - 1} \cdot \TreeCost{\ST}{\TV}.
    \end{align*}
    
    Thus, $\TreeCost{\STP}{\TV} \leq (1+\frac{1}{\Ceiling{k/2} -1}) \cdot \TreeCost{\ST}{\TV}$ for all $\TV \in \BV$.
\end{extraproof}

\subsection{Encoding Overhead Cost Functions on Trees}\label{ss:tree_function_encoding}
The second challenge occurs because the ``summaries'' of sets of queries become more complex.
Previously, when evaluating the cost of a target node \TV, two vectors of summary statistics --- for past queries on the left and past queries on the right --- were sufficient.
Importantly, these statistics never had to be merged.
In the case of searching on a tree, consider the case when a query \QUERY is made to a node with multiple previously queried subtrees.
To avoid the search space getting too large, all the queries from these different subtrees should now be summarized into one statistic ``associated with'' \QUERY.
Their own previous representative nodes might be at different distances from \QUERY.
We show how to produce such an encoding via sums of powers of distances from sets $Q$ of queries to ``interface'' nodes which separate $Q$ from the remaining feasible set.

Let $S \subseteq V$ be some feasible set (a subtree of \BT), and let $Q \subseteq \BV$ be the set of nodes queried previously (and thus resulting in the feasible set $S$).
Let $\IFV[1], \ldots, \IFV[s] \in S$ be the vertices of $S$ that are adjacent to some vertex in $\BV \setminus S$.
We refer to these vertices as the \emph{interface vertices} of $S$.
For each $i = 1, \ldots, s$, let $\IFS[i] \subseteq Q$ be the set of queries \QUERY that are not in the same connected component with $S$ in $\BT \setminus \SET{\IFV[i]}$, i.e., the queries whose paths to $S$ go through the interface vertex \IFV[i].
Notice that the \IFS[i] form a partition of the set of all queries $Q$.

For a set $Q$ of queries and a node $v$ (we will mostly use this notation for matching pairs $\IFS = \IFS[i], \IFV = \IFV[i]$ of a query set and its corresponding interface node) and some $m \in \SET{0, \ldots, \DB}$, we define $$\EnCo{m}{Q,v}  = \sum_{q \in Q} \Dist[\BT]{q}{v}^m.$$

Note first that the \EnCo{m}{\IFS,\IFV} are linear, in the sense that if $Q, Q'$ are disjoint, then
$\EnCo{m}{Q \cup Q',v} = \EnCo{m}{Q,v} + \EnCo{m}{Q',v}$.
Second, observe that if $Q$ is the set of all queries, then the total cost of $Q$ for a target \TV can be easily computed as
\begin{align*}
  \sum_{q \in Q} \GenCost[\BT]{q}{\TV}
  & = \sum_{q \in Q} \Cost{\Dist[\BT]{q}{\TV}} % \nonumber \\
  \; = \; \sum_{q \in Q} \sum_{m=0}^p \PolCo{m} \cdot \Dist[\BT]{q}{\TV}^m  \nonumber \\
  \; & = \; \sum_{m=0}^p \PolCo{m} \cdot \sum_{q \in Q} \Dist[\BT]{q}{\TV}^m %\nonumber \\
  \; = \; \sum_{m=0}^p \PolCo{m} \cdot \EnCo{m}{Q,\TV}. %\label{eqn:tree-query-cost}
\end{align*}
This also implies that for two sets $Q, Q'$, whenever $\EnCo{m}{Q,\TV} = \EnCo{m}{Q',\TV}$ for all $m$, the cost of \TV being the target under these two query sets is the same.
In this sense, the \EnCo{m}{Q,\TV} again serve as sketches, and capture all relevant information about $Q$, as far as costs are concerned.

Now, consider a feasible set $S$, an interface vertex \IFV[i] of $S$, and corresponding set \IFS[i].
Let $v \in S$ be a vertex, at distance $r = \Dist{v}{\IFV[i]}$ from \IFV[i].
For all $q \in Q$, the path from $q$ to $v$ must go through \IFV[i], so that $\Dist{q}{v} = \Dist{q}{\IFV[i]} + r$.
Therefore, we obtain that
\begin{align}
  \EnCo{m}{\IFS[i],v} & = \sum_{q \in \IFS[i]} \Dist[\BT]{q}{v}^m  %\nonumber \\
  \; = \; \sum_{q \in \IFS[i]} (\Dist[\BT]{q}{\IFV[i]} + r)^m % \nonumber \\
  \; = \;\sum_{q \in \IFS[i]} \sum_{j=0}^m \binom{m}{j} \Dist[\BT]{q}{\IFV[i]}^{m-j} r^j \nonumber \\
  & = \sum_{j=0}^m \binom{m}{j} r^j \cdot \sum_{q \in \IFS[i]}  \Dist[\BT]{q}{\IFV[i]}^{m-j} % \nonumber \\
  \; = \; \sum_{j=0}^m \binom{m}{j} \cdot r^j \cdot \EnCo{m-j}{\IFS[i],\IFV[i]}. \label{eqn:tree-cost-shift}
\end{align}

This implies that the coefficients \EnCo{m}{\IFS[i],v} for a vertex $v$ can be computed from those for \IFV[i], so long as \IFV[i] lies on the path from $v$ to $q$ for all $q \in \IFS[i]$.
For a given (connected) set $S$ of feasible vertices, we will therefore encode all relevant information about the possible costs of past queries $Q$ (with costs to be revealed in the future) in terms of the coefficients $(\EnCo{m}{\IFV[i]})_{i=1, \ldots, s, m=0, \ldots, \DB}$ (where the sets are not needed).
Because we will only optimize over $k$-cut STTs, we will have $s \leq k$, so the total number of coefficients for a given set $S$ is (at most) $(\DB+1) k$, and each coefficient \EnCo{m}{\IFS[i],\IFV[i]}, being the sum of at most \NumBV terms each of which is bounded by $\NumBV^{m}$, is at most $\NumBV^{m+1}$.
Thus, for any given feasible set $S$, there are at most $\prod_{j=0}^{\DB} \NumBV^{(j+1) k} = \NumBV^{(\DB+1)(\DB+2) \cdot k/2}$ possible settings of the coefficients.

\subsection{Putting it Together: A Dynamic Program for $k$-cut Search Trees} \label{ss:tree_alg}

As the centerpiece of our algorithm, we now derive the recurrence relation for computing the optimal search strategy costs for subtrees of \BT.
As proved in Lemma~\ref{t:cost_inflation}, at a cost increase of a factor $1+2/k$, we can restrict STTs to $k$-cut STTs.
We will do so here, implying that each feasible set (subtree) $S \subseteq \BV$ under consideration is separated from its complement by at most $k$ edges.

Fix such a set $S$ of feasible vertices, and let $\IFV[1], \ldots, \IFV[s] \in S$ (with $s \leq k$) be the vertices in $S$ with at least one neighbor outside $S$.
Consider an instance characterized by $S$, along with the encoding $(\EnCo{m}{\IFV[i]})_{i=1, \ldots, s, m=0, \ldots, \DB}$.
The optimum solution for this instance must make some next query $q \in S$.
If the adversary reveals that $\TV = q$, then the cost for this additional query is 0, and the total cost of the preceding queries is given by
$\sum_{i=1}^s \sum_{m=0}^p \PolCo{m} \cdot \EnCo{m}{\IFV[i]}$.

Otherwise, let $S' \neq \emptyset$ be the subtree which the adversary reveals to contain \TV, and let $v'$ be the unique neighbor of $q$ in $S'$.
Note that $v'$ is an interface vertex of $S'$.
Let $I = \Set{i \in \SET{1, \ldots, s}}{\IFV[i] \notin S'}$ be the indices\footnote{$I=\emptyset$ is possible.} of interface vertices of $S$ which lie outside of $S'$.
Then, all paths from $q \in \IFS[i]$ for $i \in I$ to any vertex in $S'$ must go through $v'$.
Let $Q' = \SET{q} \cup \bigcup_{i \in I} \IFS[i]$ be the set of all past queried nodes whose path to $S'$ must go through $v'$.
By linearity and \cref{eqn:tree-cost-shift}, for any $m$,
\begin{align*}
  \EnCo{m}{v'}
  & = \EnCo{m}{Q',v'}
  \; = \; \Dist{q}{v'}^m + \sum_{i \in I} \EnCo{m}{\IFS[i]}{v'} 
  \; = \; \Dist{q}{v'}^m + \sum_{i \in I} \sum_{j=0}^m \binom{m}{j} \cdot \Dist{\IFV[i]}{v'}^j \cdot {\EnCo{m-j}{\IFV[i]}}
\end{align*}
can be computed readily from all available information.
Then, the optimal solution cost for the given instance (with $S$) equals the optimum solution cost for the instance characterized by $S'$, $(\EnCo{m}{\IFV[i]})_{i=1, \ldots, s, i \notin I, m=0, \ldots, \DB}$, and $(\EnCo{m}{v'})_{m=0, \ldots, \DB}$.

The adversary's response will maximize the cost; thus, the optimum cost is the maximum cost between the first case $\TV = q$, and any of the possible responses $v'$ which are neighbors of $q$ and inside $S$, thus resulting in a non-empty set $S'$.
In turn, the algorithm tries all candidate queries $q \in S$ such that all possible resulting new sets $S'$ have boundary size at most $k$. This can be trivially checked in linear time using BFS; Lemma 7 of \cite{berendsohn2020stt} shows how to enumerate \emph{all} such candidate queries in linear time.

We thus obtain a recurrence, which can be implemented in a bottom-up fashion by increasing size of the feasible sets $S$.
There are at most $O(\NumBV^k)$ feasible sets of vertices, and for each, there are at most $\NumBV^{(\DB+1)(\DB+2) \cdot k/2}$ parameter settings for the \EnCo{m}{v} to consider.
Thus, the table has size $O(\NumBV^{k + (\DB+1)(\DB+2) \cdot k/2})$.
Each entry is computed as a minimum over at most \NumBV choices of the next query $q$, and the adversary's response as a maximum over at most $k \leq \NumBV$ possible responses.\footnote{The computation of the (constant number of) parameters for the lookup contributes at most poly-logarithmic terms, and takes constant time in a model of constant-time arithmetic.}
Thus, the total running time is $O(\NumBV^{2+(\DB^2/2+3\DB/2+2) \cdot k})$, which is polynomial in \NumBV for constant \DB and $k$.

To obtain a $(1+\epsilon)$-approximation to the true optimal STT cost (without restrictions on the cut size), we need to choose $k \geq \Ceiling{2/\epsilon} \geq 1+2/\epsilon$. Thus, a $(1+\epsilon)$-approximation can be obtained in time $O(\NumBV^{2+(\DB^2/2+3\DB/2+2) \cdot (1+2/\epsilon)})$.
This completes the proof of Theorem~\ref{t:tree}.

\begin{remark}
  While Theorem~\ref{t:tree} extends the result of Theorem~\ref{t:path} to the case of trees, the techniques we used (encoding the total cost for the queries in \SEQ with sums of powers of distances from interface vertices) cannot be used directly to extend the result of Theorem~\ref{t:multivar_path}.
  The reason is that the result for position-dependent costs crucially relies on the ability to identify the vertices of the path with integers $\SET{1, \dots, n}$, which is no longer the case for trees.
  In extending Theorem~\ref{t:path} to Theorem~\ref{t:tree}, we are able to bypass this difficulty because the cost functions depended only on the distance to the target.
\end{remark}

\section{Conclusion} \label{s:conclusion}
We introduced a model of target search with distance-dependent costs, which captures scenarios in which a principal aims to find the correct setting of a parameter, by experimenting with different settings.
The feedback is whether the setting is too high or too low, but the incurred cost (which is not observed) increases in the distance from the true value.
We first showed that when the cost function is symmetric (overshooting vs.~undershooting by the same amount leads to the same cost), standard Binary Search is a 4-approximation.
We then showed that when the cost functions are bounded-degree polynomials, the optimal strategy can be found in polynomial time.
The ideas extend to give a PTAS when the target is a node in a tree, and directional information is given (while the cost still increases in the distance).

In our work, we work in an adversarial model. A natural alternative is to assume that the target is drawn from a known distribution.
Indeed, for the line, this problem was studied in \cite{karp2000congestion}. Because the expected cost of any query can be explicitly evaluated under a known distribution, a straightforward $O(\NumBV^3)$ Dynamic Program gives the optimal solution on the line for arbitrary cost functions.
  Combining this idea with $k$-cut STTs straightforwardly leads to a PTAS for arbitrary monotone cost functions on trees. In fact, because the expected cost of a query can be explicitly evaluated, the proof extends to arbitrary ``path monotone'' cost functions \GenCost{\QUERY}{\TV}, i.e., functions which for any path originating at \QUERY are non-decreasing in the nodes of the path. The details are straightforward.

Our work raises a wealth of questions for future work.
Perhaps most immediately, for what functions $f$ can the problem be solved optimally on the line in polynomial time?
The restriction to polynomials of constant degree is essential for the application of our techniques.
Indeed, in \cref{s:impossibility}, we give evidence that the same ideas cannot be extended to other classes of cost functions. 
One may hope to obtain approximation results for general functions, e.g., by approximating them with polynomials, and then applying the dynamic programming ideas.
However, it appears that the applicability of such ideas is limited, as implicitly, not only the function, but also its derivatives, need to be well approximated.

For trees, the natural question is whether any guarantees can be obtained for arbitrary (or at least more general) monotone cost functions.
For example, can a constant-factor approximation be obtained for arbitrary (or convex) cost functions?
On the opposite end of the spectrum, as we mentioned previously, we do not know if the optimization problem is NP-hard to solve exactly for trees in this setting.
One may hope that the generalization of Binary Search to trees --- always query a node such that all resulting components have size at most $\NumBV/2$ --- might give approximation guarantees. That this is not the case can be seen by considering a path of length $\ell$ with a very large star (more than $\ell$ leaves) at one end. Generalized binary search will first query the center of the star, but if the cost function increases steeply around $\ell/2$, it would be much better to first perform a standard binary search on the path.

A very natural direction to further extend the model is by considering graph structures beyond trees.
Binary Search for general graphs was introduced in \cite{BinarySearch} and studied in a series of recent papers
\cite{dereniowski2022randomgraphs,dereniowski2021graphnoisy,dereniowski2018graphnoisyframework,krishnamurthy2021graph}.
In this model, the target is an unknown node in a given undirected (possibly edge-weighted) graph with $n$ nodes.
In response to a query $\Query{i} \neq \TV$, the algorithm learns a neighbor of \Query{i} lying on a shortest path from \Query{i} to \TV.
As shown in \cite{OnlineLearning}, Binary Search in graphs gives a natural framework for online learning in which the algorithm proposes a structure (such as a binary classifier, ranking, or clustering) and in return receives information about a mistake in its proposal (such as a misclassified data point, an out-of-order pair, or points wrongly clustered together or separated).
\cite{BinarySearch} showed that even in this general model, a target can be found in at most $\log_2 (n)$ queries.
This generalization to the general graph setting finds motivation in a variety of applications.
One such application is recommender systems, specifically ones where the content to be served can be embedded in a graph structure \cite{Wang2021GraphLB}.
Suppose that the information of the user's preferred content is specified by a node \TV within the graph.
The system may then learn about the user's preferred content through, for example, serving the user content \QUERY.
Through the user's interactions with \QUERY, the system learns about the user's preferences, which can be modeled as the direction to \TV from \QUERY.
However, serving the user \QUERY may impose a cost which increases in the distance between \QUERY and \TV.
For example, repeatedly serving content that the user dislikes may cause them to stop using the service altogether with some probability.
Hence, such a scenario fits well into the generalization of binary search that we consider.

With these potential motivations in mind, in combination with the applications given in \cite{OnlineLearning}, it is natural that the cost of a proposal should increase with the distance in $G$ from the proposal to the correct node. This differs from the objective studied in previous work, which is the ``standard'' Binary Search objective of minimizing the worst-case total number of queries, though \cite{BinarySearch} also established hardness results for the case when nodes have known query costs.
This naturally motivates extending a study of our non-uniform cost model to the case of arbitrary undirected graphs. We ask: can any non-trivial approximation guarantees (or hardness) be achieved for interesting classes of cost functions?
Even for the cost function that is identically 1 (i.e., the goal is to minimize the number of queries), minimizing the number of queries is not possible in polynomial time under the ETH.
However, no approximation hardness or non-trivial positive approximation results (beyond the trivial $O(\log n)$) are known; obtaining such results for more general cost functions would be of great interest.

Another generalization of standard binary search that has been extensively studied is the case when the answers received by the algorithm could be erroneous.
Indeed, optimal algorithms for this problem have been studied under different error models, on the line, on trees, and on general graphs. \cite{benor:hassidim:noisy-binary-search,dereniowski2021graphnoisy,dereniowski2018graphnoisyframework,feige:raghavan:peleg:upfal:noisy,karp:kleinberg:noisy}.
Combining our non-uniform cost model with errors in the answers the algorithm receives is a natural direction for future work.

Finally, it would be interesting to explicitly characterize the optimal search costs for various ``natural'' cost functions, in particular, for the linear function $\Cost{x} = x$. As mentioned in \cref{s:strategy-examples}, our experiments suggest that the cost converges to $\approx 0.6245\dots \cdot \NumBV$; while we have been able to derive a set of recursive equations that would yield the exact approximation constant as a solution, we have not been able to obtain a solution at this point.

\bibliographystyle{plainurl}

\newpage

\appendix

\section{Specific Search Strategies: Illustrative Examples} \label{s:strategy-examples}
\label{s:line-examples}

In this section, we illustrate the types of non-trivial strategies discovered by the Dynamic Program on the line.
We specifically study two natural cost functions: linear costs ($\Cost{x} = x$), and the regret from single item pricing, which has.
$\GenCost{\QUERY}{\TV} = \begin{cases} \TV & \text{ if } \TV < \QUERY\\ \TV - \QUERY & \text{ if } \QUERY \leq \TV \end{cases}$.
The examples illustrate the differences between standard Binary Search and the intricacies of adapting to specific cost functions.

\cref{f:opt_path_linear_cost} shows the optimal search strategy $\Cost{x} = x$ and $\NumBV = 10$; it has a cost of 6.
By comparison, Binary Search incurs a cost of 8.
The optimal strategy does better than Binary Search by querying vertices in the parts of the path that are ``less explored;'' that is, biasing queries more strongly towards one side of the path the more queries have been made on the opposite side.
This is to guard against multiple queries each incurring large costs.
Asymptotically, based on experiments\footnote{Although the empirical value of $0.6245\dots$ can be found by solving instances of increasing size, there is also a recurrence involving a sequence of bounds on certain rational functions tied to the query structure made by the optimal search strategy.
This sequence can be used to derive the value $0.6245\dots$ more efficiently.}, we observe that the cost of the optimal search strategy on a path of length \NumBV tends to $\approx 0.6245 \NumBV$, contrasting with the cost of Binary Search which tends to $\NumBV$.

\begin{figure}[htb]
    \centering
    \begin{tikzpicture}[every node/.style={circle,draw},level 1/.style={sibling distance=25mm},level 2/.style={sibling distance=10mm}]
        \node {5}
        child {node {2}
            child {node {1}}
            child {node {3}
                child {node {4}}
            }
        } 
        child {node {9}
            child {node {7}
                child {node {6}}
                child {node {8}}
            }
            child {node {10}}
        };
    \end{tikzpicture}
    \caption{The optimal search strategy on a path of 10 vertices with $\Cost{x} = x$. \label{f:opt_path_linear_cost}}
\end{figure}
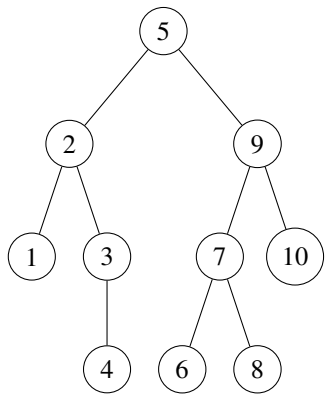

\cref{fig:opt_path_asymmetric} shows the optimal search strategy for the case where $\GenCost{\QUERY}{\TV} = \TV$ if $\TV<\QUERY$ and $\GenCost{\QUERY}{\TV} = \TV - \QUERY$ if $\QUERY \leq \TV$. The optimal strategy has a cost of 17, compared to Binary Search which achieves a cost of 23. 
Observe that this strategy performs essentially linear search in smaller intervals; overall, its structure resembles the strategy given in \cite{kleinberg2003oppa}, in that it performs ``linear-style'' search to more quickly identify the right ``range,'' then again use ``linear-style'' search to narrow within that range.

\begin{figure}[htb]
    \centering
    \begin{tikzpicture}[every node/.style={circle,draw},level 1/.style={sibling distance=35mm},level 2/.style={sibling distance=25mm}]
        \node{12}
        child {node {8}
                child {node {6}
                    child {node {4}
                        child {node {3} 
                            child{node{2} 
                                child{node {1}
                            }
                        }
                    }
                        child {node {5}}
                    }
                    child {node {7}}
                }
                child {node {9} 
                    child {node{10} 
                        child {node {11}}
                    }
                }
            }
        child {node {15}
            child {node {13}
                child {node {14}}
            }
            child {node {16}
                child {node {17}
                    child {node {18}
                        child {node {19}}
                    }
                }
            }
        };
    \end{tikzpicture}
    \caption{The optimal search strategy on the path on 20 vertices with $\GenCost{\QUERY}{\TV} = \TV$ if $\TV < \QUERY$ and $\GenCost{\QUERY}{\TV} = \TV-\QUERY$ if $\TV \geq \QUERY$. \label{fig:opt_path_asymmetric}}
\end{figure}
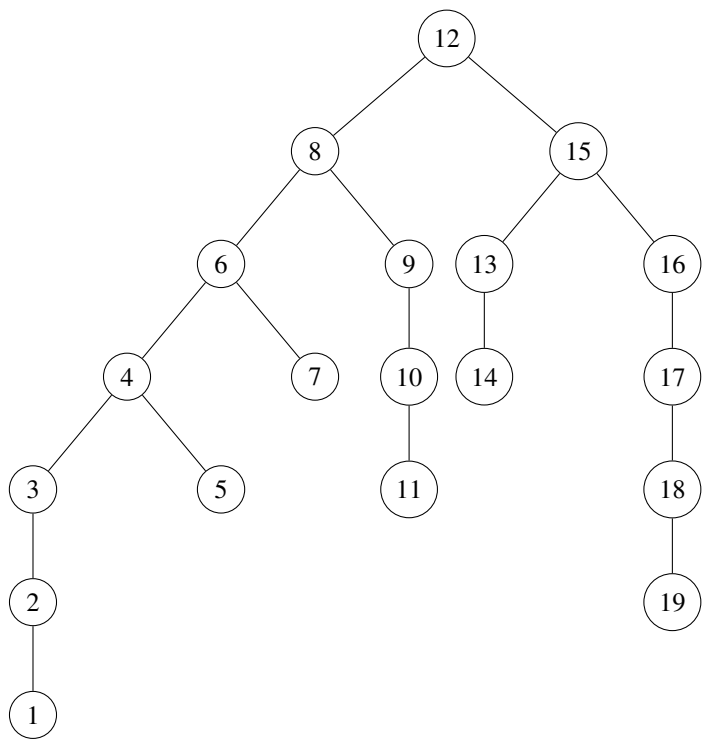

\section{Limits of the Dynamic Programming Approach} \label{s:impossibility}
Here, we provide some strong heuristical arguments for the fact that our dynamic programming approach cannot be extended beyond polynomials \COST of bounded degree.
Recall that the key idea to enable a dynamic programming approach was to choose a succinct encoding of the function \Cost[Q]{\TV} which is the cost that will have retroactively been incurred by the queries in $Q$ when the target is revealed to be \TV.

In our encoding of \COST[Q], we are using a basis of monomials and keeping the coefficients of all monomials for the past queries.
More generally, one could consider an encoding via a basis of functions $B = \SET{g_1, g_2, \dots, g_m}$, such that it is possible to write the cost function as

\begin{align*}
  \Cost[Q]{\TV} & = \sum_{q \in Q} \Cost{q - \TV}
                  \; = \; \sum_{i=1}^m c_i g_i(\TV) \text{ for all } \TV.
\end{align*}
Then, the coefficients $c_i$ would comprise the encoding of $Q$.
In our DP approaches, we specifically used the basis $B = \{1, x, x^2, \dots, x^{\DB}\}$ of monomials; the following example with $\Cost{x} = x^3$ demonstrates the desired property of allowing an encoding:
\begin{align*}
    \Cost{x - 2} & = (x-2)^3 \; = \; x^3 -6x^2 +12x -8 \in \text{Span}(\SET{1, x, x^2, \dots, x^{\DB}}).
\end{align*}
In particular, it is essential that
\begin{enumerate}
    \item \label{item:lin-comb-property}the function $\Cost{q-\TV}$ (of \TV) can be written as a $\mathbb{Z}$-linear combination of $B$ for arbitrary $q \in \SET{1, \dots, \NumBV}$;
    \item \label{item:monotonicity-property} $h$ be monotone, and thus the $g_i$ be monotone, too; and
    \item \label{item:polynomial-time-property} the extension of our methods to the basis of functions $B$ run in polynomial time.
\end{enumerate}
We prove that for any basis $B = \{g_1, \dots, g_m\}$ satisfying these properties and any function $h$ that is a linear combination of elements of $B$, $h$ must be a polynomial. Our proofs involve techniques from discrete calculus, for which we refer the reader to \cite{jordan1965calculus}.

For $i \in \mathbb{Z}$, let $E_i$ denote the \emph{shift operator} which acts on functions as $(E_i f)(x) = f(x+i)$.
Fix a finite set of functions $B = \SET{g_1, \dots, g_m}$, and denote the vector space (over $\mathbb{R}$) of functions spanned by linear combinations of elements of $B$ by $\text{Span}(B)$. Observe that Property~(\ref{item:lin-comb-property}) implies that $\text{Span}(B)$ must be closed under shift operators.
In particular, it must be closed under $E_1$.
With this in mind, fix a $g_j \in B$. Since $\text{Span}(B)$ is closed under $E_1$, the discrete derivative $\Delta g_j := g_j(x+1) - g_j(x) = E_1 g_j - g_j$ is in $\text{Span}(B)$.
Thus, since $\Delta$ is a linear operator, closure of $\text{Span}(B)$ under $E_1$ in turn implies closure under the \emph{difference operator} $\Delta$.

Now, consider the set $D := \{g_j, \Delta g_j, \Delta^2 g_j, \dots\}$. Since $\text{Span}(D)$ is a subspace of $\text{Span}(B)$, and $\text{Span}(B)$ is at most $m$-dimensional, the dimension of $\text{Span}(D)$ is at most $m$.
Thus,
the set $\{g_j, \Delta g_j, \Delta^2 g_j, \dots, \Delta^m g_j\}$ is linearly dependent, meaning that there exist constants $c_0, \dots, c_m$ (not all zero) such that $\sum_{i=0}^m c_i \Delta^i g_j = 0$.

We express the difference operators in terms of shift operators with the identity 
$$(\Delta^k g)(x) = \left[(-1)^k \sum_{\ell = 1}^k \binom{k}{\ell} \cdot (-E_1)^{k-\ell}g\right](x) = \left[\sum_{\ell = 1}^k \binom{k}{\ell} \cdot (-1)^{\ell} E_{k-\ell} g\right](x),$$ where we used that $E_1^j = E_j$.
Notice that the identity expresses each $\Delta^i$ as a polynomial of shift operators with constant coefficients, and that $E_1^j = E_j$.
Substituting this expression into the equation $\sum_{i=0}^m c_i \Delta^i g_j = 0$ results in a homogeneous linear difference equation in $g_j$ with constant coefficients, the solutions of which are expressible in the form
\begin{align*}%\label{eqn:ode-sol-form}
    g_j(x) & = \sum_{r=1}^M p_r(x) \cdot e^{\alpha_r x} \cdot (\cos(\beta_r x) + i \sin (\beta_r x))
\end{align*}
for some $M > 0$, $\alpha_r, \beta_r \in \mathbb{R}$,  and polynomial $p_r$ (refer to p.~557 of \cite{jordan1965calculus}).

Since Property~\ref{item:monotonicity-property} requires $g_j$ to be monotone, we must have $\beta_i = 0$ for all $i$.
Now suppose that there is some $g = g_j$ which has a corresponding $\alpha_r \neq 0$, and write $a = e^{\alpha_r}$.
We will show that even the set $\Set{h(1)}{h \in \text{Span}(B), x \in \SET{1, \ldots, \NumBV}}$ (i.e., evaluating all possible functions at $x=1$) has size exponential in \NumBV, meaning that the values of possible functions will not admit a concise (polynomial) encoding.
We first consider the easier case $a > 1$. Write $z = \Ceiling{\log_a 2}$. Then, $E_z g(x) = (x+z)^p \cdot a^{x+z} \geq x^p \cdot 2 \cdot a^x = 2g(x)$ for all $x$.
Let $\ell = \Floor{(\NumBV-1)/z}$, and consider coefficient vectors $\bm{c} \in \SET{0,1}^{\ell}$.
Now consider the integer linear combinations $h_{\bm{c}} = \sum_{j=1}^{\ell} c_j \cdot E_{j\cdot z} g$ for all such $\bm{c}$.
They are all in the integer span of $B$ by our preceding arguments.
Because $E_{(j+1) z}(x) \geq 2 E_{j z} (x)$ for all $x$, we get that whenever $\bm{c} \neq \bm{c'}$, the corresopnding values $h_{\bm{c}}(1) \neq h_{\bm{c'}}(1)$.
Because there are $2^{\ell} = 2^{\Omega(\NumBV)}$ such values, we obtain exponentially many values in the image of the span.

The essence of the argument is the same for $a < 1$, with only a minor difference in the detail.
Here, we will ensure that $E_z g(x) \leq g(x)/2$ for all $x$.
Solving that $(x+z)^p a^{x+z} \leq x^p a^x/2$, or equivalently $\left( \frac{x+z}{x} \right)^p \leq a^{-z}/2$ we see that it is enough to choose $z$ to satisfy $z^p \leq a^{-z}/2$. Then, we see that $z = \Omega(\log \NumBV)$ is sufficient.
The rest of the argument is essentially the same, showing that for each integer linear combination $\bm{c}$ of the form described above, the value $h_{\bm{c}}(1)$ is different. Thus, we get $2^{\ell} = 2^{\Omega(\NumBV/\log \NumBV)}$ values in the image of the span.

Thus, the assumption that there is a polynomial sized encoding rules out having any function $g_j$ with any exponential factor, meaning that all $g_j$ must be polynomials.

\end{document}